\documentclass[a4paper,11pt]{amsart}
\usepackage[utf8]{inputenc}
\usepackage[T1]{fontenc}
\usepackage{lmodern}
\usepackage{latexsym}
\usepackage{amsmath}
\usepackage{amsfonts}
\usepackage{amssymb}
\usepackage{amsthm}
\usepackage{verbatim}
\usepackage{enumitem}
\usepackage{calc}
\usepackage{tikz}
\usetikzlibrary{arrows,matrix,decorations,arrows,shapes}
\usetikzlibrary{calc}
\usetikzlibrary{decorations.markings}
\usetikzlibrary{shapes.geometric}
\usetikzlibrary{shapes.misc, positioning}
\usetikzlibrary{patterns}
\usepackage[ruled,titlenumbered,linesnumbered,noend,norelsize]{algorithm2e}
\DontPrintSemicolon \SetNlSty{footnotesize}{}{}
\IncMargin{0.5em} \SetNlSkip{1em}
\SetKwIF{If}{ElseIf}{Else}{if}{then}{else if}{else}{endif}
\usepackage{algorithmic}

\usepackage[unicode=true,pdftex]{hyperref}
\hypersetup{colorlinks=true,linkcolor=black,citecolor=black,filecolor=black,urlcolor=black,pdfauthor={Bartłomiej Bosek, Tomasz Krawczyk},pdftitle={On-line partitioning of width w posets into w\textasciicircum{}O(log log w) chains.},pdfsubject={Computer Science},pdfkeywords={partially ordered set, poset, on-line chain partition, regular poset, first-fit}, pdfproducer={LaTeX},pdfcreator={pdfLaTeX}}
\usepackage{bookmark}
\bookmarksetup{open,numbered,depth=5}
\usepackage{amsaddr}
\usepackage{fullpage}

\newtheorem{theorem}{Theorem}
\numberwithin{theorem}{section}
\newtheorem{lemma}[theorem]{Lemma}

\newtheorem{proposition}[theorem]{Proposition}

\makeatletter
\@ifpackagelater{tikz}{2013/12/01}{\def\AtanTwo(#1,#2){atan2({#1},{#2})}}{\def\AtanTwo(#1,#2){atan2({#2},{#1})}}
\makeatother

\newcommand{\convexpath}[2]{
[ 
create hullnodes/.code={
\global\edef\namelist{#1}
\foreach [count=\counter] \nodename in \namelist {
\global\edef\numberofnodes{\counter}
\node at (\nodename) [draw=none,name=hullnode\counter] {};
}
\node at (hullnode\numberofnodes) [name=hullnode0,draw=none] {};
\pgfmathtruncatemacro\lastnumber{\numberofnodes+1}
\node at (hullnode1) [name=hullnode\lastnumber,draw=none] {};
},
create hullnodes
]
($(hullnode1)!#2!-90:(hullnode0)$)
\foreach [
evaluate=\currentnode as \previousnode using \currentnode-1,
evaluate=\currentnode as \nextnode using \currentnode+1
] \currentnode in {1,...,\numberofnodes} {
-- ($(hullnode\currentnode)!#2!-90:(hullnode\previousnode)$)
let \p1 = ($(hullnode\currentnode)!#2!-90:(hullnode\previousnode) - (hullnode\currentnode)$),
\n1 = {\AtanTwo(\y1,\x1)},
\p2 = ($(hullnode\currentnode)!#2!90:(hullnode\nextnode) - (hullnode\currentnode)$),
\n2 = {\AtanTwo(\y2,\x2)},
\n{delta} = {-Mod(\n1-\n2,360)}
in 
{arc [start angle=\n1, delta angle=\n{delta}, radius=#2]}
}
-- cycle
}
\tikzstyle{wvertex}=[circle, line width=0.25mm, draw, fill=white, inner sep=0pt, minimum width=1.75mm]\tikzstyle{bvertex}=[circle, line width=0.25mm, draw, fill=black, inner sep=0pt, minimum width=1.75mm]\tikzstyle{gvertex}=[circle, line width=0.25mm, color=gray, draw, fill=gray, inner sep=0pt, minimum width=1.75mm]

\tikzstyle{edge}=[draw,line width=0.25mm,-]\tikzstyle{gedge}=[draw=gray,line width=0.25mm,-]\tikzstyle{bedge}=[draw,line width=0.65mm,-]\tikzstyle{cedge}=[bedge,draw=black!10!red]\tikzstyle{medge}=[bedge,draw=black!10!blue]\tikzstyle{bdedge}=[bedge,dash pattern=on 2mm off 1mm]\tikzstyle{dedge}=[draw,line width=0.25mm,dash pattern=on 2mm off 1mm]\tikzstyle{iedge}=[draw=black!40,line width=1mm,loosely dotted]

\tikzstyle{node}=[draw=black!10,fill=black!10]\tikzstyle{nodeborder}=[draw=black!50]\tikzstyle{nodebordered}=[draw=black!50,fill=black!10]

\tikzstyle{border}=[draw,line width=0.1mm,-]

\newlength{\hdist}
\newlength{\vdist}
\newlength{\antInner}
\newlength{\antinner}
\newlength{\nodeinner}
\newlength{\nodesecinner}
\newlength{\boxwidth}
\newlength{\boxshift}
\newlength{\boxheight}
\newlength{\innerboxheight}
\newlength{\innermargin}
\newlength{\roundedcorners}
\newlength{\treeshift}
\newlength{\vertexdist}
\tikzstyle{box}=[draw, thick, shape=rectangle, minimum width=\boxwidth, minimum height=\boxheight, anchor=center]

\colorlet{WhiteGray}{black!20}

\tikzstyle{active}=[text=black, rounded corners=0.1\boxwidth, draw=none, fill, color=WhiteGray, shape=rectangle, minimum width=\boxwidth-\innermargin, minimum height=\boxheight-\innermargin, anchor=center, path picture={
\node[color=black, pattern=north east lines] at (path picture bounding box.center) {};
}]

\tikzstyle{vital}=[text=black, rounded corners=0.1\boxwidth, fill=WhiteGray, shape=rectangle, minimum width=\boxwidth-\innermargin, minimum height=\boxheight-\innermargin, anchor=center]

\tikzstyle{nonvital}=[text=black, draw=black, thin, rounded corners=0.1\boxwidth, fill=white, shape=rectangle, minimum width=\boxwidth-\innermargin, minimum height=\boxheight-\innermargin, anchor=center, inner sep=0,outer sep=0]

\newcommand{\activenode}[4]{
\node [active, minimum width=#1-\innermargin, minimum height=#2-\innermargin, rounded corners=0.1#2+1mm] at (#3) {};
\node [fill=WhiteGray] at (#3) {#4};
}

\newcommand{\activenodee}[3]{
\node [active, minimum width=#1-\innermargin, minimum height=#2-\innermargin, rounded corners=0.1#2+1mm] at (#3) {};
}

\newcommand{\vitalnode}[4]{
\node [vital, minimum width=#1-\innermargin, minimum height=#2-\innermargin, rounded corners=0.1#2+1mm] at (#3) {};
\node at (#3) {#4};
}

\newcommand{\nonvitalnode}[4]{
\node [nonvital, minimum width=#1-\innermargin, minimum height=#2-\innermargin, rounded corners=0.1#2+1mm] at (#3) {#4};
}

\newcounter{Pi-equation}
\newenvironment{Pi-tags}{\setcounter{parentequation}{\value{equation}}\setcounter{equation}{\value{Pi-equation}}\ignorespaces}{\setcounter{Pi-equation}{\value{equation}}\setcounter{equation}{\value{parentequation}}\ignorespacesafterend}

\DeclareMathOperator{\val}{val}
\DeclareMathOperator{\regval}{rval}

\DeclareMathOperator{\width}{w}

\DeclareMathOperator{\Int}{Int}
\DeclareMathOperator{\surplus}{s}

\DeclareMathSymbol{\upset}{\mathclose}{symbols}{34}
\DeclareMathSymbol{\upseteq}{\mathclose}{symbols}{42}
\DeclareMathSymbol{\downset}{\mathclose}{symbols}{35}
\DeclareMathSymbol{\downseteq}{\mathclose}{symbols}{43}
\DeclareMathOperator{\subnode}{\vdash}

\let\leq\leqslant
\let\geq\geqslant

\makeatletter
\begin{document}

\title{On-line partitioning of width $w$ posets into $w^{O(\log{\log{w}})}$ chains}

\author{Bart{\l}omiej Bosek$^*$ \and Tomasz Krawczyk$^\dagger$}
\email{bosek@tcs.uj.edu.pl}

\email{krawczyk@tcs.uj.edu.pl}

\thanks{$^*$Research of this author is partially supported by Polish National Science Center (NCN) grant
2013/11/D/ST6/03100.}

\thanks{$^\dagger$Research of this author is partially supported by Polish National Science Center (NCN) grant 2015/17/B/ST6/01873.}

\address{Theoretical Computer Science Department, Faculty of Mathematics and Computer Science, Jagiellonian University in Krak\'{o}w, ul. {\L}ajsiewicza 6, Krak\'{o}w 30-348, Poland.}

\date{\today}

\begin{abstract}
An on-line chain partitioning algorithm receives the elements of a poset one at a time, 
and when an element is received, irrevocably assigns it to one of the chains.
In this paper, we present an on-line algorithm that partitions posets of width $w$ into $w^{O(\log{\log{w}})}$ chains.
This improves over previously best known algorithms using $w^{O(\log{w})}$ chains by Bosek and Krawczyk and by 
Bosek, Kierstead, Krawczyk, Matecki, and Smith.
Our algorithm runs in $w^{O(\sqrt{w})}n$ time, where $w$ is the width and $n$ is the size of a presented poset.
\end{abstract}

\newcommand{\sep}{,}

\maketitle

\section{Introduction}
An \emph{on-line chain partitioning algorithm} is a deterministic algorithm that receives
a poset $(P,{\leq})$ in the order of its elements $x_1,\ldots,x_n$ and constructs a chain partition 
of $(P,\leq)$ such that a chain to which $x_i$ is assigned depends solely on the poset $(P,\leq)$ restricted to $\{x_1,\ldots,x_i\}$
and on the chains to which $x_1,\ldots,x_{i-1}$ were assigned.
This formalizes the scenario in which the algorithm receives the vertices of $(P,{\leq})$ from an adversary one
at a round, and when a vertex is received, irrevocably assigns it to one of the chains.

The efficiency of an on-line chain partitioning algorithm is usually measured in terms of
the width of the presented poset.
By Dilworth's theorem, every poset of width $w$ can be partitioned into $w$~chains, and such a
partition can be computed off-line by a polynomial time algorithm~\cite{Ful56}.
In the on-line setting, the situation is more complex. 
In particular, it is not always possible
to partition width $w$ posets into $w$ chains -- see Figure \ref{fig_game_N} for an example.

\newlength{\vmargin}
\setlength{\vmargin}{1.2cm}
\newlength{\hmargin}
\setlength{\hmargin}{0.5cm}
\newlength{\borderheight}
\setlength{\borderheight}{2\hmargin+\vdist}

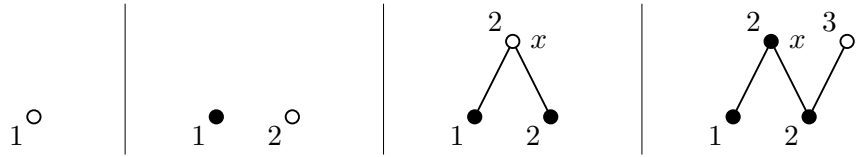
\begin{figure}[htbp!]
\begin{center}
\begin{tikzpicture}

\coordinate (a0) at (0,0);
\coordinate (a1) at ($(a0)+(\vmargin,\hmargin)$);

\coordinate (b0) at ($(a1 |- 0,0)+(\vmargin,0)$); 
\coordinate (b1) at ($(b0)+(a1)$);
\coordinate (b2) at ($(b1)+(\hdist,0)$);

\coordinate (c0) at ($(b2 |- 0,0)+(\vmargin,0)$); 
\coordinate (c1) at ($(c0)+(a1)$);
\coordinate (c2) at ($(c1)+(\hdist,0)$);
\coordinate (c3) at ($(c1)+(0.5\hdist,\vdist)$);

\coordinate (d0) at ($(c2 |- 0,0)+(\vmargin,0)$); 
\coordinate (d1) at ($(d0)+(a1)$);
\coordinate (d2) at ($(d1)+(\hdist,0)$);
\coordinate (d3) at ($(d1)+(0.5\hdist,\vdist)$);
\coordinate (d4) at ($(d3)+(\hdist,0)$);

\draw (a1) node[wvertex] (p1) {} node[below left] {$1$};

\path[border] (b0) -- ($(b0)+(0,\borderheight)$);

\draw (b1) node[bvertex] {} node[below left] {$1$};
\draw (b2) node[wvertex] {} node[below left] {$2$};

\path[border] (c0) -- ($(c0)+(0,\borderheight)$);

\path[edge] (c1) -- (c3);
\path[edge] (c2) -- (c3);
\draw (c1) node[bvertex] {} node[below left] {$1$};
\draw (c2) node[bvertex] {} node[below left] {$2$};
\draw (c3) node[wvertex] {} node[above left] {$2$} node[right=1mm] {$x$};

\path[border] (d0) -- ($(d0)+(0,\borderheight)$);

\path[edge] (d1) -- (d3);
\path[edge] (d2) -- (d3);
\path[edge] (d2) -- (d4);
\draw (d1) node[bvertex] {} node[below left] {$1$};
\draw (d2) node[bvertex] {} node[below left] {$2$};
\draw (d3) node[bvertex] {} node[above left] {$2$} node[right=1mm] {$x$};
\draw (d4) node[wvertex] {} node[above left] {$3$};

\end{tikzpicture}
\end{center}
\caption{Forcing $3$ chains on posets of width $2$.
If in the $3$-rd round an algorithm assigns $x$ to the $3$-rd chain, we are done; 
otherwise the algorithm is forced to use the $3$-rd chain in the next step.}
\label{fig_game_N}
\end{figure}

Let $\val(w)$ be the smallest $k$ for which there is an on-line algorithm that partitions posets of width $w$ 
into at most $k$ chains.
The first question one may ask is whether $\val(w)$ is bounded, i.e., whether there is an on-line algorithm that partitions posets of width $w$ into a bounded number of chains.
This question was posed by Schmerl in 1970's and at that time it was not clear whether it has a positive answer.

So far, the exact value of $\val(w)$ is known only for $w \leq 2$.
Obviously, $\val(1)=1$.
Kierstead \cite{Kie81} proved that $5 \leq \val(2) \leq 6$.
Felsner \cite{Fel97} devised an on-line algorithm that partitions posets of width~$2$ into at most $5$ chains,
showing that $\val(2) = 5$.
Kierstead \cite{Kie81} proved that $\val(3) \geq 9$ and Bosek \cite{Bos08} showed that $\val(3) \leq 16$.

Kierstead was the first to affirmatively answer Schmerl's question.
In 1981 he devised an on-line algorithm that uses exponentially many chains.
\begin{theorem}[\cite{Kie81}]
$\val(w) \leq (5^w-1)/4$.
\end{theorem}

Nearly 30 years later, the authors of this paper presented the first on-line algorithm that uses subexponentially many chains.
\begin{theorem}[\cite{BK10,BK15}]
\label{thm:alg_Bosek_Krawczyk}
$\val(w) \leq w^{13\log{w}}$.
\end{theorem}
The algorithm proving Theorem \ref{thm:alg_Bosek_Krawczyk} is quite involved.
In \cite{BKKMS18} Bosek, Kierstead, Krawczyk,
Matecki, and Smith presented another subexponential algorithm, easier to implement and analyze.
Moreover, it provides a slightly better upper bound.
\begin{theorem}[\cite{BKKMS18}]
\label{thm:alg_Bosek_Kierstead_Krawczyk_Matecki_Smith}
$\val(w) \leq w^{6.5\log{w}+7}$.
\end{theorem}
The paper \cite{BKKMS18} also shows that the techniques used to prove Theorem \ref{thm:alg_Bosek_Kierstead_Krawczyk_Matecki_Smith}
can not give a bound better than $w^{O(\log w)}$.
However, the main result of this paper proves the following:
\begin{theorem}
\label{thm:main_result}
There is an on-line algorithm that partitions posets of width $w$ into at most $w^{O(\log{\log{w}})}$ chains.
That is, $\val(w) \leq w^{O(\log{\log{w}})}$. 
\end{theorem}
The algorithm presented in this work extends the ideas from \cite{BK10,BK15}.
In Section \ref{sec:summary} we explain in details the differences between these two approaches
and we highlight the novelties of the current one.
Moreover, we discuss some issues related to the running time of our algorithm.

On the other hand, the following lower bounds on $\val(w)$ were given.
In the 1980's Kierstead \cite{Kie81} showed $\val(w) \geq 4w-3$ and
Szemer\'{e}di proved $\val(w) \geq \binom{w+1}{2}$ (see \cite{BFKKMM12, Kie86} for proofs).
Since then, the Szemer\'{e}di's lower bound was improved only by a multiplicative constant factor and the current record is $\val(w) \geq (2-o(1))\binom{w+1}{2}$~\cite{BFKKMM12}.

In 1981 Kierstead \cite{Kie81} asked whether $\val(w)$ is bounded by a polynomial in $w$.
This question still remains open and is considered to be central in this research area.

\subsection{Variants of the on-line chain partitioning problem}
Meanwhile, researchers have considered different variants of the on-line chain partitioning problem;
usually by narrowing the class of considered posets or restricting the way the posets are presented.
We refer the reader to the survey paper \cite{BFKKMM12} that provides an overview of the main research lines 
in this area.
Nevertheless, in this short subsection we provide two theorems that are used to prove the main result of the paper.

In 1995 Felsner introduced a variant of the on-line chain partitioning problem in which 
every element of a poset is maximal at the time it is presented.
The posets presented according to this limitation are called \emph{up-growing}.
This variant of the on-line chain partitioning problem was completely solved by Felsner.
\begin{theorem}[\cite{Fel97}]
\label{thm:Felsner}
There is an on-line algorithm that partitions up-growing posets of width $w$ into $\binom{w+1}{2}$ chains.
Moreover, this algorithm is best possible.
\end{theorem}

\emph{First-Fit} is an on-line chain partitioning algorithm that identifies chains with natural numbers and
assigns every new vertex of the presented poset into the lowest possible chain.
Kierstead \cite{Kie86} showed that First-Fit uses an unbounded number of chains when partitioning posets of width~$2$.
Nevertheless, there are classes of posets in which First-Fit proves to be efficient.
Probably the most prominent among them are \emph{interval posets}.
A poset $(P,{\leq})$ is \emph{interval} if every vertex $x \in P$ can be represented 
by an open interval $\mu(x)$ in $\mathbb{R}$ such that $x \leq y$ iff $\mu(x)$ is to the left of $\mu(y)$.
In a series of papers \cite{CS88,Kie88,KQ95,KST16,NB08,PR05} it is proved that First-Fit uses $O(w)$ chains in partitioning interval posets of width $w$ 
and that the constant hidden behind the big~$O$ notation is at most~$8$ and at least~$5$.
Given a poset $Q$, a poset $P$ is \emph{$Q$-free} if it does not contain an induced poset isomorphic to $Q$.
Fishburn \cite{Fis70} showed that the class of interval posets coincide with the class of $(\underline{2}+\underline{2})$-free posets,
where $(\underline{k}+\underline{k})$ is a poset consisting of two incomparable chains of height $k$.
Bosek, Krawczyk, and Szczypka~\cite{BKS10} proved that First-Fit is also efficient in the class of $(\underline{k}+\underline{k})$-free posets, where $k$ is any fixed natural number.
Precisely, they showed that First-Fit uses $O(kw^2)$ chains in partitioning $(\underline{k}+\underline{k})$-free posets of width~$w$.
Subsequently, Milans and Joret \cite{JM11} improved this bound to $O(k^2w)$, 
and Dujmovi\'{c}, Joret, and Wood proved that:
\begin{theorem}[\cite{DJW12}]
\label{thm:first_fit_two_long_incomparable_chains}
First-Fit partitions $(\underline{k}+\underline{k})$-free posets of width $w$ into at most $8kw$ chains.
Moreover, this is asymptotically best possible.
\end{theorem}

\section{Notation}
Let $(P,{\leq})$ be a poset. 
The \emph{closed upset of $A \subseteq P$}, denoted by $A\upseteq$, is the set $\{ y : x \leq y \text{ for some $x \in A$}\}$, 
and the \emph{closed downset of $A \subseteq P$}, denoted by $A\downseteq$, is the set $\{ y : y \leq x \text{ for some $x \in A$}\}$.
If $A = \{a\}$, we write $a\upseteq$ and $a\downseteq$ instead of $\{a\}\upseteq$ and $\{a\}\downseteq$.
If it does not lead to ambiguity, the \emph{subposet of $(P,{\leq})$ induced by a set $U \subseteq P$} is denoted by $(U,{\leq})$.
Let $u \in P$.
If $u\downseteq = P$, then $u$ is \emph{maximum} in $(P,{\leq})$.
If $u\upseteq = P$, then $u$ is \textit{minimum} in $(P,{\leq})$.
If $u \leq x$ yields $x = u$ for $x \in P$, $u$ is \emph{maximal} in $(P,{\leq})$.
If $x \leq u$ yields $x = u$ for $x \in P$, $u$ is \emph{minimal} in $(P,{\leq})$.

Let $A,B$ be two maximal antichains in $(P,\leq)$. 
We write $A \sqsubseteq B$ if $A \subseteq B\downseteq$ or equivalently $B \subseteq A\upseteq$.
We write $A \sqsubset B$ if $A \sqsubseteq B$ and $A \neq B$.

A poset $(P,{\leq})$ is \emph{bipartite} if the set $P$ can be partitioned into two disjoint maximal antichains $A,B$ such that
$A \sqsubset B$. 
Such a poset is denoted by $(A,B,{\leq})$.

In what follows we use the following convention: if $A,B$ are two disjoint maximum (maximal with respect to the size) 
antichains in $(P,{\leq})$ such that $A \sqsubset B$,
by $(A,B,{\leq})$ we denote the bipartite poset induced by the set $A \cup B$ in $(P,{\leq})$.
Such a poset is often treated as a~bipartite graph $(A,B,{<})$ with 
the bipartition classes $A$ and $B$ and the edges $(a < b)$, where $a \in A$, $b \in B$ are such that $a < b$.

\section{Regular posets}
\label{sec:regular_posets}
In~\cite{BK15} the on-line chain partitioning problem has been reduced to the problem of on-line chain partitioning of regular posets.
A \emph{regular poset} is a triple $(P,\leq,A_1 \ldots A_k)$, where $A_1 \ldots A_k$ is a sequence of maximum antichains of $P$, called 
the \emph{presentation order} of $(P,{\leq})$, that satisfies the following conditions: 
\begin{enumerate}
 \item $A_1,\ldots,A_k$ forms a partition of $(P,\leq)$,
 \item $(\{A_1,\ldots,A_k\}, {\sqsubseteq})$ is a linear order with $A_1$ and $A_2$ as its minimum and maximum,
 \item The antichains $A_1$ and $A_2$ induce a complete bipartite poset in $(P,{\leq})$,
 \item For every $t \in [2,k]$ and every two consecutive antichains $A_p \sqsubset A_s$ in $(\{A_1,\ldots,A_t\}, {\sqsubseteq})$, 
 the bipartite poset $(A_p,A_s,{<})$ is \emph{regular}, which means that any comparability edge in $(A_p,A_s,{<})$ is
extendable to a perfect matching in $(A_p,A_s,{<})$.
\end{enumerate}
The task of an on-line chain partitioning algorithm of a regular poset $(P,{\leq},A_1 \ldots A_k)$ remains the same;
chains to which the elements of $A_t$ are assigned depend only on the poset $(\bigcup_{i=1}^t A_i, {\leq})$ and the chains to which the elements of $\bigcup_{i=1}^{t-1} A_i$ were assigned.

Loosely speaking, a regular poset $(P,{\leq},A_1 \ldots A_k)$ is presented to an on-line algorithm in rounds.
In the $t$-th round an antichain $A_t$ is introduced and the algorithm assigns irrevocably the elements of $A_t$ to chains.
The antichains $A_1,A_2$ presented in the first two rounds form a complete bipartite poset $(A_1,A_2,{<})$.
Just before introducing an antichain $A_t$ for $t \geq 3$, 
the antichains $A_1,\ldots,A_{t-1}$ presented earlier form a partition of the presented posets into $\sqsubset$-levels, 
and the poset spanned between every two neighboring levels is regular, which means that each edge of this poset is extendable to a perfect matching.
In the $t$-th round an antichain $A_t$ is introduced between some two neighboring levels $A_p$ and $A_s$
and regular posets $(A_p,A_t,{<})$ and $(A_t,A_s,{<})$ are introduced in place of the regular poset $(A_p,A_s,{<})$.
A regular poset of width $4$ is shown in Figure~\ref{fig:regular_poset}.

\begin{figure}[htbp!]
\begin{center}
\begin{tikzpicture}

\begin{scope}[shift={(0\hdist,0\vdist)}]

\def\m{4}
\foreach \j in {1,...,\m} {
\draw (\j\hdist,1\vdist) node[wvertex] (A2\j) {};
\draw (\j\hdist,0\vdist) node[bvertex] (A1\j) {};
}

\draw[draw,dashed] \convexpath{A21,A2\m}{\antinner};
\draw[draw,dashed] \convexpath{A11,A1\m}{\antinner};

\foreach \j in {1,...,\m} {
\foreach \k in {1,...,\m} {
\path[edge] (A1\j) -- (A2\k);
}
}

\node [right] at ($(A2\m)+(\antinner,0)$) {$A_{2}$};
\node [right] at ($(A1\m)+(\antinner,0)$) {$A_{1}$};

\end{scope}

\begin{scope}[shift={(5\hdist,-0.5\vdist)}]

\def\m{4}
\foreach \j in {1,...,\m} {
\draw (\j\hdist,2\vdist) node[bvertex] (A2\j) {};
\draw (\j\hdist,1\vdist) node[wvertex] (A3\j) {};
\draw (\j\hdist,0\vdist) node[bvertex] (A1\j) {};
}

\draw[draw,dashed] \convexpath{A21,A2\m}{\antinner};
\draw[draw,dashed] \convexpath{A31,A3\m}{\antinner};
\draw[draw,dashed] \convexpath{A11,A1\m}{\antinner};

\draw[edge] (A11)--(A31)-- (A21);
\draw[edge] (A12)--(A32)-- (A22);
\draw[edge] (A13)--(A33)-- (A23);
\draw[edge] (A14)--(A34)-- (A24);

\draw[edge] (A11)--(A32);
\draw[edge] (A12)--(A31);
\draw[edge] (A13)--(A34);
\draw[edge] (A14)--(A33);

\draw[edge] (A31)--(A22);
\draw[edge] (A31)--(A23);
\draw[edge] (A32)--(A21);
\draw[edge] (A32)--(A23);
\draw[edge] (A33)--(A22);

\draw[edge, dashed] (A11) -- (A24);
\draw[edge, dashed] (A12) to[out=60,in=-145] (A24);

\draw[edge, dashed] (A14) -- (A21);
\draw[edge, dashed] (A13) to[out=120,in=-35] (A21);

\node [right] at ($(A2\m)+(\antinner,0)$) {$A_{2}$};
\node [right] at ($(A3\m)+(\antinner,0)$) {$A_{3}$};
\node [right] at ($(A1\m)+(\antinner,0)$) {$A_{1}$};

\end{scope}

\begin{scope}[shift={(10\hdist,-1\vdist)}]

\def\m{4}
\foreach \j in {1,...,\m} {
\draw (\j\hdist,3\vdist) node[bvertex] (A2\j) {};
\draw (\j\hdist,2\vdist) node[wvertex] (A4\j) {};
\draw (\j\hdist,1\vdist) node[bvertex] (A3\j) {};
\draw (\j\hdist,0\vdist) node[bvertex] (A1\j) {};
}

\draw[draw,dashed] \convexpath{A21,A2\m}{\antinner};
\draw[draw,dashed] \convexpath{A41,A4\m}{\antinner};
\draw[draw,dashed] \convexpath{A31,A3\m}{\antinner};
\draw[draw,dashed] \convexpath{A11,A1\m}{\antinner};

\draw[edge] (A11)--(A31)--(A41) -- (A21);
\draw[edge] (A12)--(A32)--(A42) -- (A22);
\draw[edge] (A13)--(A33)--(A43) -- (A23);
\draw[edge] (A14)--(A34)--(A44) -- (A24);

\draw[edge] (A11)--(A32);
\draw[edge] (A12)--(A31);
\draw[edge] (A13)--(A34);
\draw[edge] (A14)--(A33);

\draw[edge] (A31)--(A42);
\draw[edge] (A32)--(A41);

\draw[edge] (A42)--(A23);
\draw[edge] (A43)--(A22);

\draw[edge, dashed] (A11) to[out=30,in=-120] (A24);
\draw[edge, dashed] (A12) to[out=50,in=-120] (A24);

\draw[edge, dashed] (A14) to[out=150,in=-60] (A21);
\draw[edge, dashed] (A13) to[out=130,in=-60] (A21);

\node [right] at ($(A2\m)+(\antinner,0)$) {$A_{2}$};
\node [right] at ($(A4\m)+(\antinner,0)$) {$A_{4}$};
\node [right] at ($(A3\m)+(\antinner,0)$) {$A_{3}$};
\node [right] at ($(A1\m)+(\antinner,0)$) {$A_{1}$};

\end{scope}

\end{tikzpicture}
\end{center}
\caption[]{The antichains $A_2,A_3,A_4$ introduced in the $2$-nd, $3$-rd, and $4$-th round in the presentation of a regular poset $(P,{\leq},A_1\ldots A_4)$ of width $4$.}
\label{fig:regular_poset}
\end{figure}
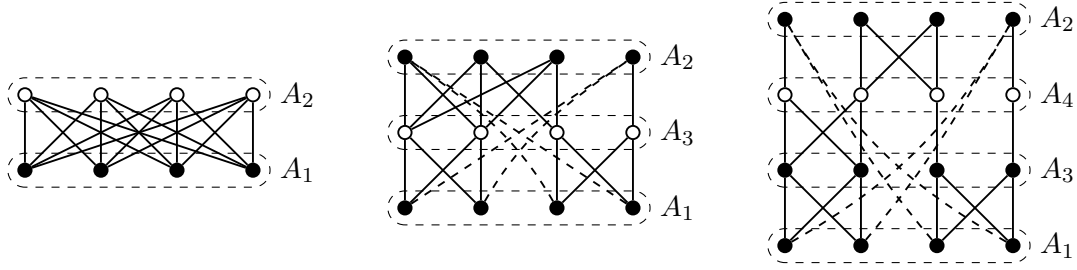

Let $\regval(w)$ be the minimum $k$ for which there is an on-line algorithm that partitions regular posets of width $w$ into at most $k$ chains.
Clearly, we have $\regval(w) \leq \val(w)$.
In \cite{BKKMS18, BK15} it is proved that $\val(w)$ can also be bounded from above in terms of~$\regval(w)$.
\begin{lemma}[\cite{BKKMS18, BK15}]
\label{lem:reg}
$\val(w) \leq w \cdot \regval(w)$.
\end{lemma}

\section{Nodes of a regular poset}

Let $(P,{\leq},A_1 \ldots A_k)$ be a regular poset.
A bipartite poset $(X,Y,{<})$ is a \emph{node} in $(P,{\leq},A_1 \ldots A_k)$
if $X \cup Y$ is a connected component in the bipartite graph $(A_p,A_s,{<})$, where $X \subseteq A_p$, $Y \subseteq A_s$, and
$A_p,A_s$ are two antichains that are consecutive at some stage of the presentation of $(P,{\leq},A_1 \ldots A_k)$.
Since $(A_p,A_s,{<})$ is regular, $(X,Y,{<})$ is also regular and we have $|X|=|Y| = \text{the width of $(X,Y,{<})$}$.

The \emph{characteristics} of a node $N = (X, Y, {<})$ is a pair $(\width(N), \surplus(N))$, where: 
\begin{itemize}
 \item $\width(N)$ is the width of~$N$, which equals to $|X|=|Y|$ as $N$ is regular, 
\item $\surplus(N)$ is the \emph{surplus} of $N$ defined as the
largest $k$ such that for all non-empty $A \subseteq X$ we have 
$|A\upseteq{} \cap Y| \geq \min\{|A|+k,|Y|\}$.
\end{itemize}
For $N$ being a complete bipartite poset the condition is true for every $k$ and we put $\surplus(N) = \infty$.
Note that the surplus of every node $N$ is at least $1$ as $(X,Y,{<})$ is connected and regular.

The next proposition shows that the surplus of $N=(X,Y,{<})$ can be equivalently defined
with respect to the downsets of the subsets of $Y$.

\begin{proposition}
\label{prop:nodes}
Let $N=(X,Y,{<})$ be a node in $(P,{\leq},A_1 \ldots A_k)$.
For every $k \in \mathbb{N}$, the following two sentences are equivalent:
\begin{enumerate}
\item \label{prop:node_surplus_definition_a} For every non-empty set $A \subseteq X$ we have $|A\upseteq{} \cap Y| \geq \min\{|A|+k, |Y|\}$.
\item \label{prop:node_surplus_definition_b} For every non-empty set $B \subseteq Y$ we have $|B\downseteq{} \cap X| \geq \min\{|B|+k, |X|\}$.
\end{enumerate}
\end{proposition}

\begin{proof}
We show that \eqref{prop:node_surplus_definition_a} implies \eqref{prop:node_surplus_definition_b}; the other direction is proved analogically.
Note that the proof will be finished if only we show the following claim:
if $t \in \mathbb{N}$ is such that $|A\upseteq{} \cap Y| \geq |A|+t$ for every non-empty subset $A \subseteq X$ of size at most $|X|-t$, 
then $|B\downseteq{} \cap X| \geq |B|+t$ for every non-empty set $B \subseteq Y$ of size at most $|Y|-t$.
Let $B \subset Y$ be a non-empty set of size at most $|Y| - t$.
Fix a subset $C$ of $B\downseteq \cap X$ as follows:
if $|B\downseteq{} \cap X| \geq t$, let $C$ be a subset of $B\downseteq{} \cap X$ of size $t$; otherwise let $C = B\downseteq{} \cap X$.
Let $D$ be any set of size $|C|$ in $Y \setminus B$.
We show that there is a perfect matching between $X \setminus C$ and $Y \setminus D$ in $(X \setminus C, Y \setminus D, {<})$.
If this holds, we deduce first that $C$ must have had $t$ elements and that the set $B\downseteq \cap X$ has size at least $|B|+t$.
This will complete the proof of our claim.
To show that $(X \setminus C, Y \setminus D, {<})$ has a perfect matching, choose a nonempty set $E$ in $X \setminus C$.
By \eqref{prop:node_surplus_definition_a}, the set $E\upseteq \cap Y$ has size at least $|E|+t$, so
$E\upseteq \cap (Y \setminus D)$ has size at least $|E|$.
So, Hall's condition holds in the bipartite poset $(X \setminus C, Y \setminus D, {<})$, which shows
that this poset contains a perfect matching.
\end{proof}

\subsection{Node tree}
Let $\Int(X,Y,{<})$ denote the set of all elements of $P$ that `lie inside' the node $(X,Y,{<})$, i.e., $\Int(X,Y,{<}) = \{x \in P: x \in X\upseteq{} \cap Y\downseteq{} \}.$
Since $(X,Y,{<})$ is regular, Dilworth Theorem asserts that the width of $(\Int(X,Y,{<}),{\leq})$ equals the width of the node $(X,Y,{<})$.

Suppose in the $t$-th round, $t \geq 3$, the antichain $A_t$ is inserted between two consecutive antichains $A_p$ and $A_s$ in $(\{A_1,\ldots,A_{t-1}\}, {\sqsubseteq})$.
The antichain $A_t$ gives rise to new nodes in bipartite posets $(A_p,A_t,{<})$ and $(A_t,A_s,{<})$.
The next proposition is a simple consequence of the fact that $(A_p,A_t,{<})$, $(A_t,A_s,{<})$, and $(A_p,A_s,{<})$ are regular.
\begin{proposition}
\label{prop:nodes_tree_structure}
For every node $M$ in $(A_p,A_t,{<})$ or in $(A_t,A_s,{<})$ there is a unique node $N$ in $(A_p,A_s,{<})$ such that 
$\Int(M) \subset \Int(N)$. 
\end{proposition}
Let $\mathcal{N}$ be the set of all nodes in $(P,{\leq},A_1 \ldots A_k)$.
If $M$ and $N$ are nodes from $\mathcal{N}$ that are in the relation as described in Proposition \ref{prop:nodes_tree_structure},
then we say $M$ is a \emph{child} of $N$, denoted $M \subnode N$.
By Proposition \ref{prop:nodes_tree_structure}, $(\mathcal{N}, {\subnode})$ is a rooted tree with the root node $(A_1,A_2,{<})$.
We call the tree $(\mathcal{N},{\subnode})$ a \emph{node tree} of $(P,{\leq},A_1 \ldots A_k)$.
See Figure \ref{fig:regular_poset_tree} for an example.

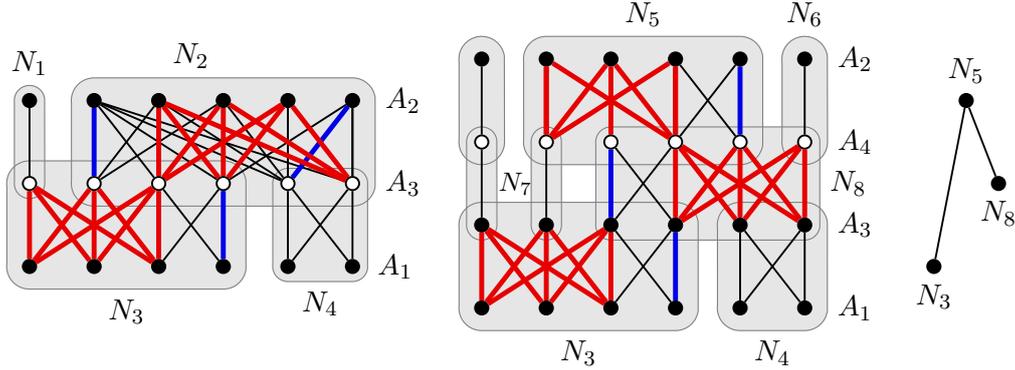
\begin{figure}[htbp!]
\begin{center}
\begin{tikzpicture}

\setlength{\hdist}{8.5mm}
\setlength{\vdist}{11mm}

\setlength{\nodeinner}{2mm}
\setlength{\nodesecinner}{3mm}

\newlength{\shiftPict}
\setlength{\shiftPict}{8\hdist}

\begin{scope}[shift={(0\hdist,0\vdist)}]

\def\n{3}
\def\m{6}
\foreach \i in {1,...,\n} {
\foreach \j in {1,...,\m} {
\coordinate (A\i\j) at (\j\hdist,\i\vdist);
}
}

\draw[node] \convexpath{A11,A21,A24,A14}{\nodesecinner};
\draw[node] \convexpath{A15,A25,A26,A16}{\nodeinner};
\draw[node] \convexpath{A21,A31}{\nodeinner};
\draw[node] \convexpath{A22,A32,A36,A26}{\nodesecinner};

\draw[nodeborder] \convexpath{A11,A21,A24,A14}{\nodesecinner};
\draw[nodeborder] \convexpath{A15,A25,A26,A16}{\nodeinner};
\draw[nodeborder] \convexpath{A21,A31}{\nodeinner};
\draw[nodeborder] \convexpath{A22,A32,A36,A26}{\nodesecinner};

\foreach \j in {1,...,\m} {\draw (A1\j) node[bvertex] (a1\j) {};}
\foreach \j in {1,...,\m} {\draw (A2\j) node[wvertex] (a2\j) {};}
\foreach \j in {1,...,\m} {\draw (A3\j) node[bvertex] (a3\j) {};}

\foreach \j in {1,...,\m} {
\path[edge] (a1\j) -- (a2\j) -- (a3\j);
}

\path[cedge] (a11) -- (a21);
\path[cedge] (a11) -- (a22);
\path[cedge] (a11) -- (a23);
\path[cedge] (a12) -- (a21);
\path[cedge] (a12) -- (a22);
\path[cedge] (a12) -- (a23);
\path[cedge] (a13) -- (a21);
\path[cedge] (a13) -- (a22);
\path[cedge] (a13) -- (a23);
\path[edge](a13) -- (a24);
\path[edge](a14) -- (a23);

\path[medge] (a14) -- (a24);

\path[edge](a15) -- (a26);
\path[edge](a16) -- (a25);

\path[edge](a22) -- (a33);
\path[edge](a22) -- (a34);
\path[edge](a23) -- (a32);
\path[edge](a23) -- (a33);
\path[edge](a23) -- (a34);
\path[edge](a24) -- (a32);
\path[edge](a24) -- (a33);
\path[edge](a24) -- (a34);
\path[edge](a24) -- (a35);
\path[edge](a24) -- (a36);
\path[edge](a25) -- (a32);
\path[edge](a25) -- (a33);
\path[edge](a25) -- (a34);
\path[edge](a25) -- (a35);
\path[edge](a25) -- (a36);
\path[edge](a26) -- (a32);
\path[edge](a26) -- (a33);
\path[edge](a26) -- (a34);
\path[edge](a26) -- (a35);
\path[edge](a26) -- (a36);

\path[medge](a22) -- (a32);
\path[medge](a25) -- (a36);

\path[cedge] (a23) -- (a33);
\path[cedge] (a23) -- (a34);
\path[cedge] (a23) -- (a35);
\path[cedge] (a24) -- (a33);
\path[cedge] (a24) -- (a34);
\path[cedge] (a24) -- (a35);
\path[cedge] (a26) -- (a33);
\path[cedge] (a26) -- (a34);
\path[cedge] (a26) -- (a35);

\node [right] at ($(a1\m)+(\nodeinner,0)$) {$A_{1}$};
\node [right] at ($(a2\m)+(\nodesecinner,0)$) {$A_{3}$};
\node [right] at ($(a3\m)+(\nodesecinner,0)$) {$A_{2}$};

\node [above] at ($(a31)+(0,\nodeinner)$) {$N_1$};
\node [above] at ($(barycentric cs:a33=0.5,a34=0.5)+(0,\nodesecinner)$) {$N_2$};
\node [below] at ($(barycentric cs:a12=0.5,a13=0.5)+(0,-\nodesecinner)$) {$N_3$};
\node [below] at ($(barycentric cs:a15=0.5,a16=0.5)+(0,-\nodeinner)$) {$N_4$};

\end{scope}

\begin{scope}[shift={(7\hdist,-0.5\vdist)}]

\def\n{4}
\def\m{6}
\foreach \i in {1,...,\n} {
\foreach \j in {1,...,\m} {
\coordinate (A\i\j) at (\j\hdist,\i\vdist);
}
}

\draw[node] \convexpath{A11,A21,A24,A14}{\nodesecinner};
\draw[node] \convexpath{A15,A25,A26,A16}{\nodesecinner};
\draw[node] \convexpath{A21,A31}{\nodeinner};
\draw[node] \convexpath{A22,A32}{\nodeinner};
\draw[node] \convexpath{A23,A33,A36,A26}{\nodeinner};
\draw[node] \convexpath{A31,A41}{\nodesecinner};
\draw[node] \convexpath{A32,A42,A45,A35}{\nodesecinner};
\draw[node] \convexpath{A36,A46}{\nodesecinner};

\draw[nodeborder] \convexpath{A11,A21,A24,A14}{\nodesecinner};
\draw[nodeborder] \convexpath{A15,A25,A26,A16}{\nodesecinner};
\draw[nodeborder] \convexpath{A21,A31}{\nodeinner};
\draw[nodeborder] \convexpath{A22,A32}{\nodeinner};
\draw[nodeborder] \convexpath{A23,A33,A36,A26}{\nodeinner};
\draw[nodeborder] \convexpath{A31,A41}{\nodesecinner};
\draw[nodeborder] \convexpath{A32,A42,A45,A35}{\nodesecinner};
\draw[nodeborder] \convexpath{A36,A46}{\nodesecinner};

\foreach \j in {1,...,\m} {\draw (A1\j) node[bvertex] (a1\j) {};}
\foreach \j in {1,...,\m} {\draw (A2\j) node[bvertex] (a2\j) {};}
\foreach \j in {1,...,\m} {\draw (A3\j) node[wvertex] (a3\j) {};}
\foreach \j in {1,...,\m} {\draw (A4\j) node[bvertex] (a4\j) {};}

\foreach \j in {1,...,\m} {
\path[edge] (a1\j) -- (a2\j) -- (a3\j) -- (a4\j);
}

\path[cedge] (a11) -- (a21);
\path[cedge] (a11) -- (a22);
\path[cedge] (a11) -- (a23);
\path[cedge] (a12) -- (a21);
\path[cedge] (a12) -- (a22);
\path[cedge] (a12) -- (a23);
\path[cedge] (a13) -- (a21);
\path[cedge] (a13) -- (a22);
\path[cedge] (a13) -- (a23);
\path[edge](a13) -- (a24);
\path[edge](a14) -- (a23);

\path[medge] (a14) -- (a24);

\path[edge](a15) -- (a26);
\path[edge](a16) -- (a25);

\path[edge](a23) -- (a34);
\path[edge](a24) -- (a33);
\path[edge](a24) -- (a36);
\path[edge](a25) -- (a34);
\path[edge](a25) -- (a35);
\path[edge](a25) -- (a36);

\path[medge] (a23) -- (a33);
\path[cedge] (a24) -- (a34);
\path[cedge] (a24) -- (a35);
\path[cedge] (a24) -- (a36);
\path[cedge] (a25) -- (a34);
\path[cedge] (a25) -- (a35);
\path[cedge] (a25) -- (a36);
\path[cedge] (a26) -- (a34);
\path[cedge] (a26) -- (a35);
\path[cedge] (a26) -- (a36);

\path[edge] (a34) -- (a45);
\path[edge] (a35) -- (a44);

\path[cedge] (a32) -- (a42);
\path[cedge] (a32) -- (a43);
\path[cedge] (a32) -- (a44);
\path[cedge] (a32) -- (a42);
\path[cedge] (a33) -- (a43);
\path[cedge] (a34) -- (a44);
\path[cedge] (a34) -- (a42);
\path[cedge] (a34) -- (a43);
\path[cedge] (a34) -- (a44);

\path[medge] (a35) -- (a45);

\node [right] at ($(a1\m)+(\nodesecinner,0)$) {$A_{1}$};
\node [right] at ($(a2\m)+(\nodesecinner,0)$) {$A_{3}$};
\node [right] at ($(a3\m)+(\nodesecinner,0)$) {$A_{4}$};
\node [right] at ($(a4\m)+(\nodesecinner,0)$) {$A_{2}$};

\node [below] at ($(barycentric cs:a12=0.5,a13=0.5)+(0,-\nodesecinner)$) {$N_3$};
\node [below] at ($(barycentric cs:a15=0.5,a16=0.5)+(0,-\nodesecinner)$) {$N_4$};
\node [above] at ($(barycentric cs:a43=0.5,a44=0.5)+(0,\nodesecinner)$) {$N_5$};
\node [above] at ($(a46)+(0,\nodesecinner)$) {$N_6$};
\node [left] at ($(barycentric cs:a22=0.5,a32=0.5)+(-0.5mm,0)$) {\begin{small}$N_7$\end{small}};
\node [right] at ($(barycentric cs:a26=0.5,a36=0.5)+(\nodeinner,0)$) {$N_8$};

\draw ($(\shiftPict,1.5\vdist)$) node[bvertex] (N3) {} node[below=1mm] {$N_3$};
\draw ($(\hdist+\shiftPict,2.5\vdist)$) node[bvertex] (N8) {} node[below=1mm] {$N_8$};
\draw ($(0.5\hdist+\shiftPict,3.5\vdist)$) node[bvertex] (N5) {} node[above=1mm] {$N_5$};
 
\path[edge] (N3) -- (N5);
\path[edge] (N8) -- (N5);

\end{scope}

\end{tikzpicture}
\end{center}
\caption[]{The presentation of the antichains $A_3$ and $A_4$ in a regular poset $(P,{\leq},A_1\ldots A_4)$ of width $6$. 
In the node tree: the nodes $N_1,N_2,N_3,N_4$ are the children of the node $(A_1,A_2,{<})$,
the nodes $N_5,N_6,N_7,N_8$ are the children of the node $N_2$.
The node $N_2$ has characteristics $(5,2)$ and is active: the edges of the Dilworth's clique of $N_2$ are colored red, 
a~perfect matching extending the Dilworth's clique is depicted with blue. 
The nodes $N_3,N_5,N_8$ have characteristics $(4,1)$ and all are active -- 
the partial order $(\mathcal{P}(4,1),{\leq_{(4,1)}})$ defined on the nodes $N_3,N_5,N_8$ is shown on the right.
The nodes $N_5$ and $N_8$ are added to $(\mathcal{P}(4,1),{\leq_{(4,1)}})$ at the time $A_4$ is presented, 
i.e., at the same time when $N_5$ and $N_8$ appear in the node tree.
}
\label{fig:regular_poset_tree}
\end{figure}
The next proposition shows that the pairs (width, surplus) of characteristics are weakly decreasing with respect to the
lexicographic order along root-to-leaf paths of the node tree $(\mathcal{N},{\subnode})$.
\begin{proposition}
\label{prop:characteristics_lexicographic_order}
Suppose $M$ is a descendant of $N$ in $(\mathcal{N},\subnode)$. Then:
\begin{enumerate}
 \item $\width(M) \leq \width(N)$.
 \item If $\width(M) = \width(N)$, then $\surplus(M) \leq \surplus(N)$.
\end{enumerate}
\end{proposition}
\begin{proof}
The lemma follows by the facts that the poset $(\Int(N),{\leq})$ has the width $\width(N)$ and that $\Int(M) \subset \Int(N)$.
\end{proof}

A set $\mathcal{F} \subseteq \mathcal{N}$ is \emph{ancestor-free} 
if neither node from $\mathcal{F}$ is a descendant of a node in $\mathcal{F}$.

\begin{proposition}
\label{prop:matchings}
Suppose $\mathcal{F}$ is an ancestor-free set in $(\mathcal{N}, {\subnode})$ 
and suppose a perfect matching $M(X,Y,{<})$ is fixed for every node $(X,Y,{<})$ in $\mathcal{F}$.
Let $F$ be the set of all vertices from $P$ incident to the nodes from $\mathcal{F}$,
i.e. $F = \bigcup_{(X,Y,{<}) \in \mathcal{F}} X \cup Y$.
Then, there is a chain partition of $(F,{\leq})$ into $w$ chains such that
for every node $(X,Y,{<})$ in $\mathcal{F}$ and for every edge $(a < b)$ in $M(X,Y,{<})$,
the elements $a,b$ are in the same chain of this chain partition.
\end{proposition}

\begin{proof}
Let $\mathcal{A}$ denote the set $\mathcal{F}$ extended by all nodes $N$ such that $N$ is a leaf of $(\mathcal{N},{\subnode})$
that is not a descendant of a node from $\mathcal{F}$.
Fix a perfect matching $M(N)$ for each node $N \in \mathcal{A} \setminus \mathcal{F}$.
Now, $M(N)$ is fixed for every $N$ in $\mathcal{A}$.
Let $A$ be the set of all vertices from $P$ incident to the nodes from~$\mathcal{A}$.
Clearly, $F \subseteq A$.
Let ${<_M}$ be the transitive closure of the set $\{(a, b): (a < b) \text{ is an edge from } M(N) \text{ for } N \in \mathcal{A} \}$.
Note that $(A,{\leq_M})$ is a partial order that consists of $w$ chains $C_1,\ldots,C_w$ -- see Figure \ref{fig_chains}.
Clearly, the restriction of $C_1,\ldots,C_w$ to the set $F$ yields the desired chain partition.
\end{proof}
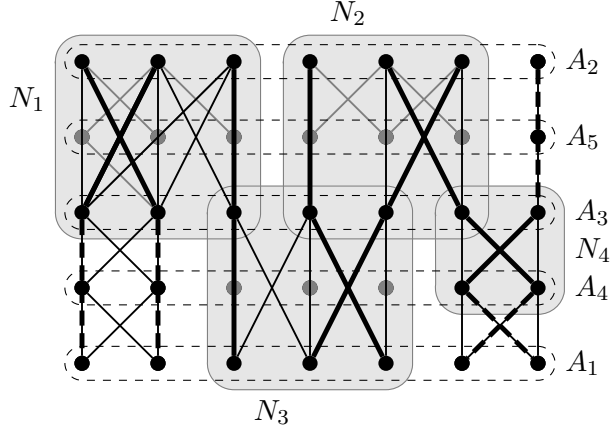
\begin{figure}[htbp!]
\begin{center}
\begin{tikzpicture}

\def\n{7}
\def\m{5}
\foreach \i in {1,...,\n} {
\foreach \j in {1,...,\m} {
\coordinate (a\j\i) at ($(\i\hdist,\j\vdist)$);
\draw ($(\i\hdist,\j\vdist)$) node[bvertex] (a\j\i) {};
}
}

\draw[node] \convexpath{a51,a53,a33,a31}{\nodeinner};
\draw[node] \convexpath{a33,a35,a15,a13}{\nodeinner};
\draw[node] \convexpath{a54,a56,a36,a34}{\nodeinner};
\draw[node] \convexpath{a36,a37,a27,a26}{\nodeinner};

\draw[nodeborder] \convexpath{a51,a53,a33,a31}{\nodeinner};
\draw[nodeborder] \convexpath{a33,a35,a15,a13}{\nodeinner};
\draw[nodeborder] \convexpath{a54,a56,a36,a34}{\nodeinner};
\draw[nodeborder] \convexpath{a36,a37,a27,a26}{\nodeinner};

\foreach \i in {1,...,\n} {
\foreach \j in {1,...,\m} {
\draw (a\j\i) node[bvertex] {};
}
}

\draw (a41) node[gvertex] {};
\draw (a42) node[gvertex] {};
\draw (a43) node[gvertex] {};
\draw (a44) node[gvertex] {};
\draw (a45) node[gvertex] {};
\draw (a46) node[gvertex] {};

\draw (a23) node[gvertex] {};
\draw (a24) node[gvertex] {};
\draw (a25) node[gvertex] {};

\node [right] at ($(a57)+(\antinner,0)$) {$A_2$};
\node [right] at ($(a47)+(\antinner,0)$) {$A_5$};
\node [right] at ($(a37)+(\nodeinner,0)$) {$A_3$};
\node [right] at ($(a27)+(\nodeinner,0)$) {$A_4$};
\node [right] at ($(a17)+(\antinner,0)$) {$A_1$};

\path[gedge] (a51) -- (a42);
\path[gedge] (a41) -- (a52);
\path[gedge] (a52) -- (a43);
\path[gedge] (a42) -- (a53);

\path[gedge] (a54) -- (a45);
\path[gedge] (a44) -- (a55);
\path[gedge] (a55) -- (a46);
\path[gedge] (a45) -- (a56);

\path[gedge] (a41) -- (a32);
\path[gedge] (a31) -- (a42);

\foreach \i in {1,...,\n} {
\path[edge] (a5\i) -- (a1\i);
}

\path[bedge] (a52) -- (a31);
\path[bedge] (a51) -- (a32);
\path[edge](a53) -- (a31);
\path[edge](a53) -- (a32);
\path[edge](a52) -- (a33);
\path[bedge] (a52) -- (a31);

\path[bedge] (a53) -- (a33);
\path[bedge] (a54) -- (a34);
\path[bedge] (a56) -- (a35);
\path[bedge] (a55) -- (a36);
\path[bdedge] (a57) -- (a37);

\path[bdedge] (a31) -- (a11);
\path[edge](a31) -- (a22);
\path[edge](a32) -- (a21);
\path[edge](a21) -- (a12);
\path[edge](a22) -- (a11);
\path[bdedge] (a32) -- (a12);

\path[bedge] (a33) -- (a13);
\path[edge]  (a33) -- (a14);
\path[edge]  (a34) -- (a13);
\path[bedge] (a34) -- (a15);
\path[bedge] (a35) -- (a14);

\path[bedge]  (a36) -- (a27);
\path[bedge]  (a37) -- (a26);
\path[edge]   (a26) -- (a17);
\path[edge]   (a27) -- (a16);
\path[bdedge] (a26) -- (a17);
\path[bdedge] (a27) -- (a16);

\foreach \j in {1,...,\m} {
\draw[draw=black,dashed] \convexpath{a\j7,a\j1}{\antinner};
}

\node [left] at ($(barycentric cs:a41=0.5,a51=0.5)+(-\nodeinner,0)$) {$N_1$};
\node [above] at ($(barycentric cs:a54=0.5,a55=0.5)+(0,\nodeinner)$) {$N_2$};
\node [below] at ($(barycentric cs:a13=0.5,a14=0.5)+(0,-\nodeinner)$) {$N_3$};
\node [right] at ($(barycentric cs:a27=0.5,a37=0.5)+(\nodeinner,0)$) {$N_4$};

\end{tikzpicture}
\end{center}
\caption{The nodes from $\mathcal{F}$ are shaded.
The matchings $M(N)$ are marked in bold lines for $N \in \mathcal{F}$ and with bold dotted lines for $N \in 
\mathcal{A} \setminus \mathcal{F}$.}
\label{fig_chains}
\end{figure}

\subsection{Edge orders}
For a set of nodes $\mathcal{F}$ from $\mathcal{N}$, we denote by $\mathcal{F}_E$ the set of all edges in the nodes from
$\mathcal{F}$.
Furthermore, we equip the set $\mathcal{F}_E$ with a partial order relation $\leq_{E}$ defined:
$$(a < b) \leq_{E} (c < d) \text { iff } b \leq c, \text{ for every } (a < b), (c < d) \in \mathcal{F}_E.$$
The poset $(\mathcal{F}_E, {\leq_E})$ is called an \emph{edge poset} on $\mathcal{F}$.
The next lemma proves a polynomial upper bound on the width of an edge poset provided it is grounded on an ancestor-free set of nodes.

\begin{lemma}
\label{lem:ancestor_free_edge_order_width}
Suppose $\mathcal{F}$ in an ancestor-free set of nodes in $(\mathcal{N}, \subnode)$.
Then, the edge poset $(\mathcal{F}_E, {\leq_E})$ has width at most $w^3$.
\end{lemma}

\begin{proof}
Let $A$ be an antichain in $(\mathcal{F}_E, {\leq_E})$ and let $\mathcal{A}$ consist of nodes that contain at least one edge from $A$.
We claim that $\mathcal{A}$ contains at most $w$ nodes. 
If this holds, $A$ contains at most $w^3$ edges as every node has at most $w^2$ edges, and the lemma follows.
To prove the claim, suppose $\mathcal{A}$ contains at least $w+1$ nodes.
For every node $N$ from $\mathcal{A}$, first we choose an edge $(a < b)$ in $N$ such that $(a < b) \in A$ and then we choose a perfect matching $M(N)$ containing the edge $(a < b)$. 
Note that such a matching exists as $N$ is regular.
Since $\mathcal{A}$ is ancestor-free, by Proposition \ref{prop:matchings}, we can partition the poset $(A,{\leq})$ into $w$ chains 
such that for every node $N \in \mathcal{A}$ and every edge $(a < b)$ in $M(N)$, the vertices $a$ and $b$ are contained in the same chain. 
Now, since the matchings $M(N)$ for $N \in \mathcal{A}$ contain at least $w+1$ edges from $A$, 
some two edges from this set share the same chain of this chain partition.
Since $\mathcal{A}$ is ancestor-free, these two edges are $\leq_{E}$-comparable.
This contradicts the fact that $A$ is an antichain in $(\mathcal{F}_E, {\leq_E})$.
\end{proof}

\section{The algorithm partitioning regular posets -- general idea}
In this section we present the main ideas used in the construction of the on-line algorithm partitioning 
regular posets of width $w$ into $w^{O(\log{\log{w}})}$ chains.
Let $(P,{\leq},A_1 \ldots A_k)$ be a regular poset of width $w$ and let $(\mathcal{N},{\subnode})$ be its node tree.
Instead of partitioning the elements of $(P,{\leq})$ into chains,
the algorithm colors on-line the edges of the nodes from $(\mathcal{N},{\subnode})$.
That is, each edge of a node $N \in \mathcal{N}$ is assigned a non-empty set (a bundle) of colors immediately after $N$ appears in the node tree $(\mathcal{N}, {\subnode})$ 
such that the following property is kept:
\begin{equation}
\label{eq:prop_edge_coloring}
\begin{array}{c}
\text{For every color $\gamma$, the endpoints of the edges that have $\gamma$ in its bundle of colors} \\
\text{form a chain in $(P,{\leq})$.}
\tag{*}
\end{array}
\end{equation}
The next step is easy.
In every round $t$, to every vertex $x \in A_t$ the algorithm assigns a color from some edge incident to $x$.
Condition \eqref{eq:prop_edge_coloring} guarantees that all points with the same color lie in one chain of $(P,{\leq})$.

Let $\lambda:\mathbb{N} \to \mathbb{N}$ be a function defined recursively such that $\lambda(1) = 1$ and
$$\lambda(w) = \underbrace{ \Bigg( 16\lfloor w^{\frac{15}{2}} \rfloor \lambda(\lfloor \sqrt{w} \rfloor) \Bigg)}_{\lambda_1(w)} \cdot \underbrace{\Bigg(2\binom{w^3+1}{2}+2 \Bigg)}_{\lambda_2(w)} \cdot \underbrace{\Bigg(w^4\Bigg(2\binom{w^3+1}{2}+2\Bigg)\lambda(\lfloor \sqrt{w} \rfloor)\Bigg)}_{\lambda_3(w)}$$
for $w \geq 2$, and let $\lambda_1, \lambda_2, \lambda_3: \mathbb{N} \to \mathbb{N}$ be as defined above.
So, $\lambda(w) \leq poly(w)\cdot \lambda^2(\sqrt{w})$, where $poly(w)$ is a polynomial of degree at most $24$.
Performing simple calculations one can check that:
$$\lambda(w) = w^{O(\log{\log{w}})}.$$
The algorithm we are going to present works such that the following invariant is kept:
\begin{equation}
\label{eq:P-property}
\begin{array}{c}
\text{If every edge of $(A_1,A_2,{<})$ is covered by a bundle of $\lambda(w)$ colors, then all remaining} \\
\text{nodes of $\mathcal{N}$ are colored with these colors such that condition \eqref{eq:prop_edge_coloring} is satisfied.}
\end{array}
\tag{P}
\end{equation}
If \eqref{eq:P-property} holds, the algorithm partitions regular posets of width $w$ into at most $w^2 \lambda(w)$ chains as $(A_1,A_2,{<})$ has $w^2$ edges.
Since the algorithm uses the reduction from Lemma \ref{lem:reg}, it partitions on-line posets of width $w$ into $w^3 \lambda(w)$ chains.
Hiding the factor $w^3$ under big $O$ notation, we conclude that the algorithm uses $w^{O(\log{\log{w})}}$ chains in partitioning width $w$ posets.
This proves the main result of the paper.

\subsection{An overview of the techniques used}
In this subsection we describe some techniques used frequently by the algorithm.
\subsubsection{Replacing chains by colors}
Suppose $N$ is a node and $\mathcal{F}$ is a set of some descendants of $N=(X,Y,<)$.
Suppose that an external procedure $\mathbb{P}$ partitions on-line 
the edge poset $(\mathcal{F}_E, {\leq}_E)$ into at most $k$ chains.
The use of the \emph{replacing-chains-by-colors} method
allows us to achieve the following goal:
if each edge of $N$ is colored by a bundle of $k$ colors, 
then, with the use of these colors, we can color the edges from $\mathcal{F}_E$ such that \eqref{eq:prop_edge_coloring}-property is kept.
To attain our goal, we match $k$ colors assigned to every edge of $N$ with $k$ chains produced by~$\mathbb{P}$.
Then, if $\mathbb{P}$ assigns an edge $(z < t)$ from $\mathcal{F}_E$ to a chain $c$,
we take an edge $(x < y)$ from~$N$ such that $x \leq z < t \leq y$ and color $(z < t)$ with a color from a bundle of $(x < y)$ matched to the chain $c$.
That such an edge $(x < y)$ exists follows from the fact that the width of $(\Int(N),{\leq})$ equals to $\width(N)$.
See Figure~\ref{fig_three_techniques} on the left for an example.

\begin{figure}[htbp!]
\begin{center}
\begin{tikzpicture}

\begin{scope}[shift={(0,0)}]

\setlength{\hdist}{10mm}
\setlength{\vdist}{6.2mm}

\def\n{7}
\foreach \i in {1,...,\n} {
\draw (\i\hdist,0) node[bvertex] (x\i) {};
\draw (\i\hdist,7\vdist) node[bvertex] (y\i) {};
\path[dedge] (x\i) -- (y\i);
}

\draw (2\hdist,1\vdist) node[bvertex] (z2) {};
\draw (3\hdist,2\vdist) node[bvertex] (z3) {};
\draw (5\hdist,3\vdist) node[bvertex] (z5) {} node[above left=0.5mm] (Z) {$z$};
\draw (6\hdist,4\vdist) node[bvertex] (z6) {} node[right=1mm] (T) {$t$};

\draw (3\hdist,3\vdist) node[bvertex] (t3) {};
\draw (4\hdist,4\vdist) node[bvertex] (t4) {};
\draw (5\hdist,5\vdist) node[bvertex] (t5) {};
\draw (6\hdist,6\vdist) node[bvertex] (t6) {};

\draw (x5) node[bvertex] {} node[left=1mm] (X) {$x$};
\draw (y6) node[bvertex] {} node[right=1mm] (Y) {$y$};

\node [left,fill=white] at ($(barycentric cs:x2=0.3,y3=0.7)+(0mm,0)$) {$1'\!,2'$};
\node [left,fill=white] at ($(barycentric cs:x3=0.2,y4=0.8)+(2mm,0)$) {$1''\!,2''$};
\node [left,fill=white] at ($(barycentric cs:x5=0.8,y6=0.2)+(2mm,0)$) {$1'''\!,2'''$};

\node [right,fill=white] at ($(barycentric cs:z2=0.5,z3=0.5)+(-0.6mm,-1.1mm)$) {$2'$};
\node [left,fill=white] at ($(barycentric cs:z5=0.5,z6=0.5)+(2mm,1mm)$) {$2'''$};
\node [left,fill=white] at ($(barycentric cs:t3=0.55,t4=0.45)+(2mm,1mm)$) {$1''$};
\node [left,fill=white] at ($(barycentric cs:t5=0.5,t6=0.5)+(2mm,1mm)$) {$1'''$};

\path[edge] (z2) -- (z3);
\path[dedge] (z3) -- (z5);
\path[edge] (z5) -- (z6);

\path[edge] (t3) -- (t4);
\path[dedge] (t4) -- (t5);
\path[edge] (t5) -- (t6);

\draw[edge] (x2) to [bend left=5] (y3);
\draw[edge] (x3) to [bend left=-7] (y4);
\draw[edge] (x5) to [bend left=-5] (y6);

\draw[draw,dashed] \convexpath{x1,x\n}{\antinner};
\draw[draw,dashed] \convexpath{y1,y\n}{\antinner};
\draw[draw] \convexpath{x1,y1,y\n,x\n}{\nodeinner};

\node [right] at ($(x\n)+(\nodeinner,0)$) {$X$};
\node [right] at ($(y\n)+(\nodeinner,0)$) {$Y$};
\draw ($(barycentric cs:x6=0.7,x7=0.3)+(0,1\vdist)$) node[fill=white] {$\Int(N)$};
\end{scope}

\begin{scope}[shift={(8.5\vdist,0\hdist)}]

\setlength{\hdist}{10mm}
\setlength{\vdist}{11mm}

\def\n{5}
\foreach \i in {1,...,\n} {
\draw (\i\hdist,0) node[bvertex] (x\i) {};
\draw (\i\hdist,1.5\vdist) node[bvertex] (z\i) {};
\draw (\i\hdist,2.5\vdist) node[bvertex] (t\i) {};
\draw (\i\hdist,4\vdist) node[bvertex] (y\i) {};
\path[dedge] (x\i) -- (z\i);
\path[gedge](z\i) -- (t\i);
\path[dedge] (t\i) -- (y\i);
}

\draw (y4) {} node[right=1mm] {$y$};
\draw (t4) {} node[right=1mm] {$t$};
\draw (z3) {} node[left=1mm] {$z$};
\draw (x3) {} node[left=1mm] {$x$};

\path[gedge](z1) -- (t2);
\path[gedge](z2) -- (t3);
\path[gedge](z4) -- (t5);
\path[gedge](z3) -- (t1);
\draw[gedge](z4) to [bend right=20] (t3);
\path[gedge](z5) -- (t4);
\draw[bedge] (y4) to [bend left=9] (x3);
\path[bedge] (z3) -- node [midway,left=0mm] {$c$} (t4);

\node [right] at ($(barycentric cs:y4=0.5,t4=0.5)+(-5mm,0)$) {$c$};

\draw[draw,dashed] \convexpath{x1,x\n}{\antinner};
\draw[draw,dashed] \convexpath{z1,z\n}{\antinner};
\draw[draw,dashed] \convexpath{t1,t\n}{\antinner};
\draw[draw,dashed] \convexpath{y1,y\n}{\antinner};
\draw[draw] \convexpath{x1,y1,y\n,x\n}{\nodeinner};

\node [right] at ($(x\n)+(\nodeinner,0)$) {$X$};
\node [right] at ($(barycentric cs:x\n=0.5,z\n=0.5)+(\nodeinner,0)$) {$M_1$};
\node [right] at ($(z\n)+(\nodeinner,0)$) {$A_p'$};
\node [right] at ($(t\n)+(\nodeinner,0)$) {$A_s'$};
\node [right] at ($(barycentric cs:y\n=0.5,t\n=0.5)+(\nodeinner,0)$) {$M_2$};
\node [right] at ($(y\n)+(\nodeinner,0)$) {$Y$};
\end{scope}

\end{tikzpicture}
\end{center}
\caption{}
\label{fig_three_techniques}
\end{figure}

\subsubsection{Projecting colors}
Suppose $N=(X,Y,{<})$ is a node and suppose $A_p \sqsubset A_s$ are two antichains
consecutive at some stage of the presentation of $(P,{\leq},A_1 \ldots A_k)$ such that the sets
$A_p' = \Int(N) \cap A_p$ and $A_s'=\Int(N) \cap A_s$ are non-empty.
Suppose $\mathcal{F}$ is a set of some nodes spanned between $A_p'$ and $A_s'$.
The use of the \emph{projecting-colors} method allows us to achieve the following goal:
if each edge is colored with one color, each with a different one, 
then, with the use of these colors, we can color
the edges from $\mathcal{F}_E$ such that \eqref{eq:prop_edge_coloring}-property is kept.
To attain our goal, we pick a perfect matching $M_1$ from $X$ to $A'_p$ in $(X, A'_p, {<})$
and a perfect matching $M_2$ from $A'_s$ to $Y$ in $(A'_s,Y, {<})$;
such matchings exist as $(\Int(N),{\leq})$ has width $\width(N)$.
Eventually, each edge $(z < t) \in \mathcal{F}_E$ receives the color
from an edge $(x < y)$ of $N$, where $x$ is matched with $z$ in $M_1$ and $t$ is matched with $y$ in $M_2$.
See Figure \ref{fig_three_techniques} on the right for an example.

\subsubsection{Shuffling colors}
Suppose $N=(X,Y,{<})$ is a complete bipartite poset and suppose
$N$ has two children $M' = (X,Z,{<})$ and $M'' = (Z,Y,{<})$ being also complete bipartite posets.
Suppose $X=\{x_0,\ldots,x_{u-1}\}$, $Y=\{y_0,\ldots,y_{u-1}\}$, and $Z=\{z_0,\ldots,z_{u-1}\}$.
The use of the \emph{shuffling-colors} method allows us to achieve the following goal:
if each edge of $N$ is colored with one color, each with a different one, then,
with the use of these colors, we can color the edges of $M'$ and $M''$ such that
\eqref{eq:prop_edge_coloring}-property is kept.
To perform this task, for every edge $(x_i < y_j)$ in $N$ we compute $k = ((i+j) \mod u)$ and we color
$(x_i < z_k)$ and $(z_k < y_j)$ with the color of the edge $(x_i < y_j)$.

\section{The details of the algorithm}
The algorithm consists of two main procedures. 
The first colors so-called active nodes in $\mathcal{N}$.
The second, which is called for every active node $N$, 
colors nodes \emph{dependent} on $N$, that is, all non-active nodes $M$
for which $N$ is the first active node on a path from $M$ to the root of $(\mathcal{N}, {\subnode})$.
We denote this set by $D(N)$.
Since the root of $(\mathcal{N},{\subnode})$ is active,
the sets $D(N)$ for active nodes $N$ form a partition of the set $\mathcal{N}$.
In the next two subsections we describe how the algorithm colors active nodes and how the algorithm colors 
nodes dependent on an active node~$N$.

\subsection{Coloring active nodes}
\label{subsec:active_nodes}
A \emph{Dilworth clique of width $k$} in a node $N=(X,Y,{<})$ is a set $\{x_1,\ldots,x_k,y_1,\ldots,y_k\}$
with $x_1,\ldots,x_k \in X$, $y_1,\ldots,y_k \in Y$, with the relation $x_i < y_j$ for $i,j \in [k]$,
such that there is an extension of edges $x_1 < y_1, \ldots, x_k < y_k$ to a perfect matching in $N$.
One can easily check that any ancestor of a node with a Dilworth's clique of width $k$ contains a Dilworth's clique of the same width.

A node $N$ with characteristics $(u,s)$ is called \emph{active} if
it contains a Dilworth clique of width $\lceil \sqrt{w} \rceil$
and has no ancestor in the node tree with the same characteristics.
For each active node $N$ a Dilworth clique is fixed once $N$ appears in the node tree.
We denote this Dilworth clique by $R(N)$.
On the set of all active nodes $\mathcal{P}(u,s)$ with characteristics $(u,s)$ 
we define a poset $(\mathcal{P}(u,s), {\leq_{(u,s)}})$
by the rule that $N <_{(u,s)} K$ iff there is a maximal $x$ in $(R(N),{\leq})$ and a minimal $y$ in $(R(K),{\leq})$ with $x \leq y$.
It is clear that $\leq_{(u,s)}$ is a partial order and that $\leq_{(u,s)}$
can be extended on every active node from $\mathcal{P}(u,s)$ immediately after it appears in the node tree -- see Figure \ref{fig:regular_poset_tree} for an example.
The next lemma proves another crucial property of $(\mathcal{P}(u,s), {\leq_{(u,s)}})$.

\begin{lemma}
The width of $(\mathcal{P}(u,s), {\leq_{(u,s)}})$ is at most $\lfloor \sqrt{w} \rfloor$.
\end{lemma}

\begin{proof}
For every node $N$ in $\mathcal{P}(u,s)$ choose a perfect matching $M(N)$ of $N$ consisting of $\lceil \sqrt{w} \rceil$ 
edges with the endpoints in $R(N)$ and $(u - \lceil \sqrt{w} \rceil)$ edges with the endpoints outside $R(N)$.
Suppose $P(u,s)$ is the set of all elements of $P$ incident to the nodes from $\mathcal{P}(u,s)$.
Since $\mathcal{P}(u,s)$ is ancestor-free, by Proposition \ref{prop:matchings}, 
there is a chain partition $\mathcal{C}$ of size $w$ of $(P(u,s),{\leq})$ 
such that for every node $N \in \mathcal{P}(u,s)$ and every edge $(a < b)$ from $M(N)$, 
the elements $a,b$ are in a same chain of $\mathcal{C}$.
Suppose to the contrary that there is an antichain $\mathcal{A}$ of size $\lfloor \sqrt{w} \rfloor + 1$ in $(\mathcal{P}(u,s),{\leq_{(u,s)}})$.
For every $N \in \mathcal{A}$, the edges from $M(N)$ that have their endpoints in $R(N)$ belong to $\lceil \sqrt{w} \rceil$ different chains of $\mathcal{C}$.
Since $|\mathcal{A}| = \lfloor \sqrt{w} \rfloor + 1$, there are two edges, say $(a < b) \in R(N)$ and $(c < d) \in R(K)$ for some $N,K \in \mathcal{A}$, that lie in the same chain of $\mathcal{C}$.
Since $\mathcal{A}$ is ancestor-free, we have either $(a < b) \leq (c < d)$ or $(c < d) \leq (a < b)$.
So, either $N \leq_{(u,s)} K$ or $K \leq_{(u,s)} N$, which contradicts the fact
that $\mathcal{A}$ is an antichain in $(\mathcal{P}(u,s), \leq_{(u,s)})$.
\end{proof}

The algorithm recursively partitions on-line poset $(\mathcal{P}(u,s),{\leq}_{(u,s)})$ into chains.
Since the width of $(\mathcal{P}(u,s),{\leq_{(u,s)}})$ is at most $\lfloor \sqrt{w} \rfloor$,
$\mathcal{P}(u,s)$ is partitioned into at most $\lfloor \sqrt{w} \rfloor^3 \lambda(\lfloor \sqrt{w} \rfloor)$ chains.
Suppose $\mathcal{L}$ is a chain in $(\mathcal{P}(u,s), \leq_{(u,s)})$ generated by the algorithm.
The next lemma shows that First-Fit can be used to partition efficiently the edge poset $(\mathcal{L}_E, {\leq_{E}})$ into chains.

\begin{lemma}
\label{lem:edge_order_for_a_chain}
Let $\mathcal{L}$ be a chain in $(\mathcal{P}(u,s),{\leq_{(u,s)}})$. 
The edge poset $(\mathcal{L}_E, {\leq_{E}})$ has width at most $w^3$ and is $(\underline{2w-2\lceil \sqrt{w}\rceil +3}+\underline{2w-2\lceil \sqrt{w} \rceil +3})$-free.
\end{lemma}

\begin{proof}
Since $\mathcal{L}$ is ancestor-free, $(\mathcal{L}_E, {\leq_{E}})$ has width at most $w^3$ by Lemma~\ref{lem:ancestor_free_edge_order_width}.

Let $w' = w - \lceil \sqrt{w} \rceil + 2$. 
Then $2w'-1 = 2w-2\lceil \sqrt{w} \rceil +3$.
To prove that the edge poset $(\mathcal{L}_E, {\leq_{E}})$ is $(\underline{2w'-1}+\underline{2w'-1})$-free, assume to the contrary that
$$(a_1 < b_1) \leq \ldots \leq (a_{2w'-1} < b_{2w'-1}) \text{ and } (c_1 < d_1) \leq \ldots \leq (c_{2w'-1} < d_{2w'-1})$$
is a $(\underline{2w'-1}+\underline{2w'-1})$ structure in $(\mathcal{L}_E, {\leq_{E}})$.
Suppose $(a_i < b_i)$ is an edge of a node $N_i$ and $(c_i < d_i)$ is an edge of a node $M_i$, for $i \in [2w'-1]$.
We claim that:
\begin{itemize}
 \item $b_1 < x$ for some minimal $x$ in $(R(N_{w'}), {\leq})$,
 \item $d_1 < x$ for some minimal $x$ in $(R(M_{w'}), {\leq})$,
 \item $y < a_{2w'-1}$ for some maximal $y$ in $(R(N_{w'}), {\leq})$,
 \item $y < c_{2w'-1}$ for some maximal $y$ in $(R(M_{w'}), {\leq})$.
\end{itemize}
See Figure \ref{fig_2w2wfree_2} for an illustration.
\begin{figure}[htbp!]
\begin{center}
\begin{tikzpicture}

\draw (0,3.5\vdist) node[bvertex] (b2){} node[right=1mm] {$b_{2w'-1}$};
\draw (0,3\vdist) node[bvertex] (a2){} node[right=1mm] {$a_{2w'-1}$};
\draw (0,1.5\vdist) node[bvertex] (n12) {};
\draw (0,1\vdist) node[bvertex] (n11) {};
\draw (0,0.5\vdist) node[bvertex] (b1){} node[right=1mm] {$b_1$};
\draw (0,0\vdist) node[bvertex] (a1){} node[right=1mm] {$a_1$};

\draw (\hdist,1.5\vdist)node[bvertex] (n22) {};
\draw (\hdist,1\vdist)node[bvertex] (n21) {};
\draw (2\hdist,1.5\vdist) node[bvertex] (n32) {};
\draw (2\hdist,1\vdist) node[bvertex] (n31) {};
\draw (2\hdist,2.5\vdist) node[bvertex] (m32) {};
\draw (2\hdist,2\vdist) node[bvertex] (m31) {};
\draw (3\hdist,2.5\vdist) node[bvertex] (m22) {};
\draw (3\hdist,2\vdist) node[bvertex] (m21) {};
\draw (4\hdist,3.5\vdist) node[bvertex] (d2){} node[right=0.15cm] {$d_{2w'-1}$};
\draw (4\hdist,3\vdist) node[bvertex] (c2){} node[right=0.15cm] {$c_{2w'-1}$};
\draw (4\hdist,2.5\vdist) node[bvertex] (m12) {};
\draw (4\hdist,2\vdist) node[bvertex] (m11) {};
\draw (4\hdist,0.5\vdist) node[bvertex] (d1){} node[right=0.15cm] {$d_1$};
\draw (4\hdist,0\vdist) node[bvertex] (c1){} node[right=0.15cm] {$c_1$};

\path[edge](a1)-- (b1);
\path[dedge] (b1)-- (n11);
\path[dedge] (n12) -- (a2);
\path[edge](a2)-- (b2);

\path[cedge] (n11) -- node [midway,left=1mm] {$R(N_{w'})$} (n12);
\path[cedge] (n11) -- (n22);
\path[cedge] (n11) -- (n32);
\path[cedge] (n21) -- (n12);
\path[cedge] (n21) -- (n22);
\path[cedge] (n21) -- (n32);
\path[cedge] (n31) -- (n12);
\path[cedge] (n31) -- (n22);
\path[cedge] (n31) -- (n32);
\path[dedge] (n32) -- (m31);
\path[cedge] (m11) -- node [midway,right=1mm] {$R(M_{w'})$} (m12);
\path[cedge] (m11) -- (m22);
\path[cedge] (m11) -- (m32);
\path[cedge] (m21) -- (m12);
\path[cedge] (m21) -- (m22);
\path[cedge] (m21) -- (m32);
\path[cedge] (m31) -- (m12);
\path[cedge] (m31) -- (m22);
\path[cedge] (m31) -- (m32);

\path[edge](c1)-- (d1);
\path[dedge] (d1)-- (m11);
\path[dedge] (m12) -- (c2);
\path[edge](c2)-- (d2);

\end{tikzpicture}
\end{center}
\caption{}
\label{fig_2w2wfree_2}
\end{figure}
First, using the above claims, we prove the statement of the lemma.
Since $\mathcal{L}$ is a chain in $(\mathcal{P}(u,s),{\leq}_{(u,s)})$, we have either
$N_{w'} \leq_{(u,s)} M_{w'}$ or $M_{w'} \leq_{(u,s)} N_{w'}$.
Suppose $N_{w'} \leq_{(u,s)} M_{w'}$. 
By the first statement of the claim, we have $b_1 < x$ for some minimal $x$ in $(R(N_{w'}), {\leq})$.
Hence, since $(R(N_{w'}),{\leq})$ is a complete bipartite poset, $b_1$ is $\leq$-below all maximal elements in $(R(N_{w'}), {\leq})$.
Similarly, by the last statement of the claim, $y < c_{2w'-1}$ for some maximal $y$ in $(R(M_{w'}), {\leq})$, 
and $c_{2w'-1}$ is $\leq$-above every minimal element in $(R(M_{w'}),{\leq})$.
Since $N_{w'} \leq_{(u,s)} M_{w'}$, some maximal element in $(R(N_{w'}), {\leq})$
is below some minimal element in $(R(M_{w'}), {\leq})$.
By transitivity of ${\leq}$ we get $(a_1 < b_1) \leq (c_{2w'-1} < d_{2w'-1})$, which contradicts the fact that the chains
$(a_1 < b_1) \leq \ldots \leq (a_{2w'-1} < b_{2w'-1})$ and $(c_1 < d_1) \leq \ldots \leq (c_{2w'-1} < d_{2w'-1})$ are $\leq_{E}$-incomparable.
Similarly, if $M_{w'} < N_{w'}$, we get $(c_1 < d_1) \leq (a_{2w'-1} < b_{2w'-1})$, 
and we reach the same conclusion as earlier.

To complete the proof of the lemma, we need to show the claim.
Due to symmetry, we prove only its first sentence.
Suppose that $a_i \in A(a_i)$ and $b_i \in A(b_i)$ for some $A(a_i), A(b_i)$ in $\{A_1,\ldots,A_k\}$.
If $b_1$ is $\leq$-below some minimal element in $(R(N_j),{\leq})$ for some $j \in [2,w']$,
then, because of $N_{j} \leq_{(u,s)} N_{w'}$, $b_1$ is $\leq$-below some minimal element in $(R(N_{w'}),{\leq})$, and the claim holds.
So, we assume that:
\begin{equation}
\label{eq:upsets_and_cliques}
\big(b_1 \upseteq \cap A(a_j)\big) \cap \big(R(N_j) \cap A(a_j)\big) = \emptyset \text{ for every } j \in [2,w'].
\end{equation}
Under this assumption we will show that:
\begin{equation}
\label{eq:upsets_1}
\big|b_1 \upseteq \cap A(b_j)\big| \geq j \text{ for every } j \in [1,w'-1].
\end{equation}
Note that \eqref{eq:upsets_1} implies:
\begin{equation}
\label{eq:upsets_2}
\big|b_1\upseteq \cap A(a_{j+1})\big| \geq j \text{ for every } j \in [1,w'-1],
\end{equation}
as either $A(b_j)=A(a_{j+1})$ or $A(b_j) \sqsubset A(a_{j+1})$ 
and there is a perfect matching between $A(b_j)$ and $A(a_{j+1})$ in $(A(b_j),A(a_{j+1}),{<})$.
However, note that \eqref{eq:upsets_2} for $j=w'-1$ yields
$\big|b_1 \upseteq \cap A(a_{w'})\big| \geq w'-1 = w - \lceil \sqrt{w} \rceil +1$, which
contradicts \eqref{eq:upsets_and_cliques} for $j=w'$ as $|R(N_{w'}) \cap A(a_{w'})| = \lceil \sqrt{w} \rceil$.

We prove \eqref{eq:upsets_1} by induction on $j$.
Clearly, \eqref{eq:upsets_1} holds for $j=1$.
Suppose \eqref{eq:upsets_1} holds for $j=i-1$ for some $i \in [2,w'-1]$; we will show \eqref{eq:upsets_1} holds for $j=i$.
By \eqref{eq:upsets_2}, we have $\big|b_1\upseteq \cap A(a_{i})\big| \geq i-1$.
Suppose $N_i = (X,Y,{<})$, for some $X \subseteq A(a_i)$ and $Y \subseteq A(b_i)$.
We split the set $b_1 \upseteq \cap A(a_i)$ into two sets:
$C = (b_1\upseteq \cap A(a_i)) \setminus X$ and $D=(b_1\upseteq \cap A(a_i)) \cap X$ -- see Figure \ref{fig_2w2wfree_1}.
\begin{figure}[htbp!]
\begin{center}
\begin{tikzpicture}

\setlength{\hdist}{13mm}
\setlength{\vdist}{13mm}

\setlength{\antInner}{2.5mm}
\setlength{\antinner}{3.5mm}
\setlength{\nodeinner}{4.5mm}

\def\n{9}
\foreach \i in {1,...,\n} {
\coordinate (A\i) at (\i\hdist,0);
\coordinate (B\i) at (\i\hdist,\vdist);
\coordinate (C\i) at (\i\hdist,2.5\vdist);
\coordinate (D\i) at (\i\hdist,3.5\vdist);
}

\draw[nodebordered] \convexpath{C5,D5,D\n,C\n}{\nodeinner};

\draw[draw=black] \convexpath{B3,B6}{\antInner};
\draw[draw=black] \convexpath{C3,C4}{\antInner};
\draw[draw=black] \convexpath{C5,C6}{\antInner};
\draw[draw=black] \convexpath{D3,D4}{\antInner};
\draw[draw=black] \convexpath{D5,D7}{\antInner};

\foreach \i in {1,...,\n} {
\draw (A\i) node[bvertex] (a\i) {};
\draw (B\i) node[bvertex] (b\i) {};
\draw (C\i) node[bvertex] (c\i) {};
\draw (D\i) node[bvertex] (d\i) {};
}

\foreach \i in {1,...,\n} {
\path[edge] (a\i) -- (b\i);
\path[dedge] (b\i) -- (c\i);
\path[edge] (c\i) -- (d\i);
}

\path[cedge] (a1) -- (b1);
\path[cedge] (a1) -- (b2);
\path[cedge] (a1) -- (b3);
\path[cedge] (a2) -- (b1);
\path[cedge] (a2) -- (b2);
\path[cedge] (a2) -- (b3);
\path[cedge] (a3) -- (b1);
\path[cedge] (a3) -- (b2);
\path[cedge] (a3) -- (b3);

\path[edge] (a3) -- (b4);
\path[edge] (a4) -- (b5);
\path[edge] (a4) -- (b3);
\path[edge] (a5) -- (b4);

\path[bedge] (a5) -- (b5);

\path[edge] (a6) -- (b7);
\path[edge] (a7) -- (b8);
\path[edge] (a8) -- (b9);
\path[edge] (a9) -- (b6);

\path[edge] (c1) -- (d2);
\path[edge] (c2) -- (d1);

\path[bedge] (c5) -- (d5);

\path[edge] (c5) -- (d6);
\path[edge] (c6) -- (d5);
\path[edge] (c6) -- (d7);
\path[edge] (c7) -- (d6);
\path[edge] (c7) -- (d8);
\path[edge] (c8) -- (d7);

\path[cedge] (c7) -- (d7);
\path[cedge] (c7) -- (d8);
\path[cedge] (c7) -- (d9);
\path[cedge] (c8) -- (d7);
\path[cedge] (c8) -- (d8);
\path[cedge] (c8) -- (d9);
\path[cedge] (c9) -- (d7);
\path[cedge] (c9) -- (d8);
\path[cedge] (c9) -- (d9);

\node [node] at (barycentric cs:c6=0,d6=0,c7=0.25,d7=0.25) {\begin{small}$N_i=(X, Y,<)$\end{small}};
\node [above] at ($(barycentric cs:d3=0.5,d4=0.5) + (0,\antinner)$) {$C\upseteq\cap A(b_i)$};
\node [above] at ($(d6) + (0,\nodeinner)$) {$D\upseteq\cap A(b_i)$};
\node [below] at ($(barycentric cs:c3=0.5,c4=0.5) + (0,-\antinner)$) {$C$};
\node [below] at ($(barycentric cs:c5=0.5,c6=0.5) + (0,-\nodeinner)$) {$D$};
\node [above,fill=white] at ($(barycentric cs:b4=0.5,b5=0.5) + (0,\antinner)$) {$b_1\upseteq\cap A(b_{i-1})$};

\node [right] at ($(a9) + (\antinner,0)$) {$A(a_{i-1})$};
\node [right] at ($(b9) + (\antinner,0)$) {$A(b_{i-1})$};
\node [right] at ($(c9) + (\nodeinner,0)$) {$A(a_{i})$};
\node [right] at ($(d9) + (\nodeinner,0)$) {$A(b_{i})$};

\node [right] at (a5) {$a_{i-1}$};
\node [right] at (b5) {$b_{i-1}$};
\node [right] at (c5) {$a_{i}$};
\node [right] at (d5) {$b_{i}$};

\draw[draw=black,dashed] \convexpath{a1,a\n}{\antinner};
\draw[draw=black,dashed] \convexpath{b1,b\n}{\antinner};
\draw[draw=black,dashed] \convexpath{c1,c\n}{\antinner};
\draw[draw=black,dashed] \convexpath{d1,d\n}{\antinner};

\end{tikzpicture}
\end{center}
\caption{}
\label{fig_2w2wfree_1}
\end{figure}
Consider the set $C\upseteq \cap A(b_i)$. 
Since $(X,Y,{<})$ is a node spanned between $A(a_i)$ and $A(b_i)$, 
the set $C\upseteq \cap A(b_i)$ is disjoint with $Y$.
Moreover, since there is a perfect matching between $A(a_i)$ and $A(b_i)$ in $(A(a_i),A(b_i),{<})$,
the set $C\upseteq \cap A(b_i)$ has at least $|C|$ elements.
Now consider the set $D$.
Note that $D$ is non-empty as $a_i$ belongs to $D$.
Note also that $D$ is strictly contained in $X$ as $D$ is disjoint with $R(N_i) \cap A(a_i)$ by \eqref{eq:upsets_and_cliques}.
Since $N_i$ has surplus at least $1$, we conclude that $D\upseteq \cap Y$ has size at least 
$|D|+1$.
So, $b_1\upseteq \cap A(b_i)$ has at least $|C|+|D|+1 \geq i$ elements.
This proves~\eqref{eq:upsets_1}.
\end{proof}

Based on Lemma \ref{lem:edge_order_for_a_chain} and Theorem \ref{thm:first_fit_two_long_incomparable_chains}, 
First-Fit partitions the edge poset $(\mathcal{L}_E, {\leq_{E}})$ into at most $8 \cdot (2w-2\lceil \sqrt{w} \rceil + 3) \cdot w^3 < 16w^4$ chains. 
Since we have at most:
\begin{itemize}
 \item $w^2$ possible characteristics $(u,s)$,
 \item $\lfloor \sqrt{w} \rfloor^3 \lambda(\lfloor \sqrt{w} \rfloor)$ chains $\mathcal{L}$ into which $(\mathcal{P}(u,s),{\leq_{(u,s)}})$ is partitioned,
 \item $16w^4$ chains into which $(\mathcal{L}_E, {\leq_{E}})$ is partitioned by First-Fit,
\end{itemize}
we conclude that the edges of the active nodes in $(\mathcal{N}, {\subnode})$ can be partitioned into 
$\lambda_1(w) = 16\lfloor w^{\frac{15}{2}} \rfloor \lambda(\lfloor \sqrt{w} \rfloor)$ edge chains.
Eventually, the use of the replacing-chains-with-colors method leads us to the following:
\begin{proposition}
\label{prop:coloring_active_nodes_summary}
If every edge of $(A_1,A_2,{<})$ is colored by a bundle of $\lambda_1(w)$ colors, then,
with the use of these colors, the algorithm can color all active nodes in $\mathcal{N}$ such that 
\eqref{eq:prop_edge_coloring}-property is satisfied.
\end{proposition}

\subsection{Coloring dependent nodes of an active node}
To color all nodes dependent on active nodes, we need to have every edge of an active node colored by $\lambda_2(w)\lambda_3(w)$ colors.
To assert such number of colors, we replace every color used for coloring the edges of active
nodes by a bundle of $\lambda_2(w)\lambda_3(w)$ colors.
By Proposition~\ref{prop:coloring_active_nodes_summary}, this requires to have $\lambda(w)=\lambda_1(w)\lambda_2(w)\lambda_3(w)$ colors on every edge of $(A_1,A_2,{<})$, as we have guaranteed.

Suppose $N$ is an active node with characteristics $(u,s)$.
The algorithm to be presented colors every node from $D(N)$ using the colors assigned to the edges of $N$.
Therefore, to guarantee \eqref{eq:prop_edge_coloring} in $(P,{\leq})$, 
it suffices to check \eqref{eq:prop_edge_coloring} locally in the poset $(\Int(N),{\leq})$.

\subsubsection{The structure of the nodes in $D(N)$}
In the set $D(N)$ we distinguish the following types of nodes:
\begin{itemize}
\item a node $M$ in $D(N)$ is \emph{recursive} if $M$ is the first node with 
the width at most $\lfloor \sqrt{w} \rfloor$ on the path from $N$ to $M$ in the node tree,
\item a node $M$ in $D(N)$ is \emph{problematic} if $\width(M) \geq \lfloor \sqrt{w} \rfloor + 1$ and 
$M$ has no Dilworth clique of width $\lceil \sqrt{w} \rceil$.
In particular, a problematic node can not be complete.
\end{itemize}
By Proposition \ref{prop:characteristics_lexicographic_order}, all nodes in $D(N)$ have characteristics lexicographically not exceeding $(u,s)$.
Suppose $Q$ is the set of all nodes from $D(N)$ with characteristics $(u,s)$.
Note that every child of $Q$ with characteristics lexicographically strictly smaller than $(u,s)$ 
is either active, problematic, or recursive.
Active children of the nodes from $Q$ are already handled.
In Subsection \ref{subsec:Q_and_their_children} we show how the algorithm colors the nodes from $Q$ and 
the problematic and recursive children of the nodes from $Q$.
In Subsection \ref{subsec:recursive_nodes} we show how the descendants of a recursive node are colored.
Eventually, in Subsection \ref{subsec:problematic_nodes} we show how the descendants of a problematic child of a node from $Q$ are colored. 

\subsubsection{The nodes from $Q$ and their recursive and problematic children}
\label{subsec:Q_and_their_children}
To describe how the algorithm colors the nodes from $Q$ and their recursive and problematic children, 
we split into two cases: $N$ is not a complete bipartite poset and $N$ is a complete bipartite poset.

Suppose $N$ is not a complete bipartite poset, i.e., $N$ has characteristic $(u,s)$ with $s < \infty$.
Under this assumption, the set $Q$ forms a non-empty path in $(\mathcal{N},{\subnode})$ starting at $N$.
Let $L$ be the last node in the path $Q$. 
That is, if $L$ has children, all of them have characteristics smaller than $(u,s)$.
All recursive and problematic children of the nodes from $Q \setminus \{L\}$ are kept in the sets $A(Q)$ and $B(Q)$.
Loosely speaking, $A(Q)$ and $B(Q)$ contain nodes that lie
`above' and `below' the path $Q$, respectively -- see Figure \ref{fig:Q_path}.
Formally, the sets $A(Q)$ and $B(Q)$ are processed as follows.
When $N$ appears in the node tree, we have $Q=N$ and we set $A(Q)=\emptyset$, and $B(Q)=\emptyset$.
Now, suppose we are at the moment when a node, say $N'$, from the path $Q$ is being split. 
If $N'$ has a child $M'$ of characteristics $(u,s)$ (i.e., $Q$ is extended by a node $M'$),
the sets $A(Q)$ and $B(Q)$ are set such that:
\begin{itemize}
 \item if $M'$ is a lower-layered child of $N'$, $A(Q)$ is extended by the upper-layered recursive and problematic children of $N'$ and $B(Q)$
 remains unchanged,
 \item if $M'$ is an upper-layered child of $N'$, $B(Q)$ is extended by the lower-layered recursive and problematic children of $N'$ and $A(Q)$
 remains unchanged.
\end{itemize}
Otherwise, if $N'$ has no children with characteristics $(u,s)$, the construction of $A(Q)$ and $B(Q)$ is finished.

Now, consider the edge posets $(A(Q)_E,{\leq_E})$ and $(B(Q)_E,{\leq_E})$ as on-line posets and note the following:
\begin{proposition}
\label{rem:A(Q)_B(Q)_path_online}
The edge posets $(A(Q)_E,{\leq_E})$ and $(B(Q)_E,{\leq_E})$ are down-growing and up-growing orders of width at most $w^3$, respectively.
\end{proposition}
\begin{proof}
It is clear that the sets $A(Q)$ and $B(Q)$ are ancestor-free.
By Lemma \ref{lem:ancestor_free_edge_order_width}, the width of $(A(Q)_E,{\leq_E})$ and $(B(Q)_E,{\leq_E})$ is at most $w^3$.
That $(A(Q)_E,{\leq_E})$ and $(B(Q)_E,{\leq_E})$ are, respectively, down-growing and up-growing follows by the construction of $A(Q)$ and $B(Q)$.
\end{proof}
We use Felsner's algorithm (see Theorem \ref{thm:Felsner}) to partition on-line each of $(A(Q)_E,{\leq_E})$ and $(B(Q)_E,{\leq_E})$ into at most $\binom{w^3+1}{2}$ edge chains.
Then, using the replacing-chains-by-colors method we achieve the following:
if every edge of $N$ is assigned a bundle of $2\binom{w^3+1}{2}$ colors,
then we can color with the colors assigned to the edges of $N$ all the edges of $A(Q)_E \cup B(Q)_E$ such that \eqref{eq:prop_edge_coloring}-property holds.

We are left with the problem of coloring the nodes from $Q$ and possible recursive and problematic children of the last node $L$ in $Q$.
For this purpose, we reserve $2$ additional colors in a bundle of every edge of $N$.
The coloring of the nodes from $Q$ is easy: we project additional colors from a parent $N'$ to its child $M'$ 
for every two consecutive nodes $N',M'$ in $Q$.
To color the recursive and problematic children of the last node $L$ in $Q$, 
we project one additional color from every edge of $L$ onto upper-layered children of $L$ and the second one onto lower-layered children of $L$.

\begin{figure}[htbp!]
\begin{center}
\begin{tikzpicture}

\coordinate (A) at (2,2);
\coordinate (B) at ($(A)+(\boxwidth+\boxshift,0)$);
\coordinate (C) at ($(B)+(\boxwidth+\boxshift,0)$);
\coordinate (D) at ($(C)+(\boxwidth+\boxshift,0)$);
\coordinate (E) at ($(D)+(\boxwidth+\boxshift,0)$);

\coordinate (B1) at ($(B)+(0,\innerboxheight*2)$);
\coordinate (B2) at ($(B)-(0,\innerboxheight/2)$);

\coordinate (C1) at ($(C)+(0,\innerboxheight*2)$);
\coordinate (C2) at (C);
\coordinate (C5) at ($(C)-(0,\innerboxheight*2)$);

\coordinate (D1) at ($(D)+(0,\innerboxheight*2)$);
\coordinate (D2) at ($(D)+(0,\innerboxheight/2)$);
\coordinate (D4) at ($(D)-(0,\innerboxheight)$);
\coordinate (D5) at ($(D)-(0,\innerboxheight*2)$);

\coordinate (E1) at ($(E)+(0,\innerboxheight*2)$);
\coordinate (E2) at ($(E)+(0,\innerboxheight)$);
\coordinate (E3) at (E);
\coordinate (E4) at ($(E)-(0,\innerboxheight)$);
\coordinate (E5) at ($(E)-(0,\innerboxheight*2)$);

\newcommand{\LevelI}[1]{
\coordinate (B11) at ($(#1)-(\boxwidth/3,0)$);
\coordinate (B12) at (#1);
\coordinate (B13) at ($(#1)+(\boxwidth/3,0)$);
\nonvitalnode{\boxwidth/3}{\innerboxheight}{B11}{$F_1$}
\activenodee{\boxwidth/3}{\innerboxheight}{B12}
\nonvitalnode{\boxwidth/3}{\innerboxheight}{B13}{$F_2$}}

\newcommand{\LevelII}[1]{
\coordinate (E21) at ($(#1)-(\boxwidth/6,0)$);
\coordinate (E22) at ($(#1)+(\boxwidth/3,0)$);
\nonvitalnode{2\boxwidth/3}{\innerboxheight}{E21}{$F_7$}
\nonvitalnode{\boxwidth/3}{\innerboxheight}{E22}{$F_8$}}

\newcommand{\LevelIII}[1]{
\coordinate (E31) at ($(#1)-(\boxwidth/3,0)$);
\coordinate (E32) at (#1);
\coordinate (E33) at ($(#1)+(\boxwidth/3,0)$);
\activenodee{\boxwidth/3}{\innerboxheight}{E31}
\activenodee{\boxwidth/3}{\innerboxheight}{E32}
\nonvitalnode{\boxwidth/3}{\innerboxheight}{E33}{$F_9$}}

\newcommand{\LevelIV}[1]{
\coordinate (D41) at ($(#1)-(\boxwidth/4,0)$);
\coordinate (D42) at ($(#1)+(\boxwidth/4,0)$);
\nonvitalnode{\boxwidth/2}{\innerboxheight}{D41}{$F_5$}
\nonvitalnode{\boxwidth/2}{\innerboxheight}{D42}{$F_6$}}

\newcommand{\LevelV}[1]{
\coordinate (C51) at ($(#1)-(\boxwidth/3,0)$);
\coordinate (C52) at (#1);
\coordinate (C53) at ($(#1)+(\boxwidth/3,0)$);
\nonvitalnode{\boxwidth/3}{\innerboxheight}{C51}{$F_3$}
\nonvitalnode{\boxwidth/3}{\innerboxheight}{C52}{$F_4$}
\activenodee{\boxwidth/3}{\innerboxheight}{C53}}

\node [box] at (A) {};
\activenode{\boxwidth}{\boxheight}{A}{$N=N_1$}

\node [box] at (B) {};
\node [box, minimum height=\innerboxheight] at (B1) {};
\node [box, minimum height=4\innerboxheight] at (B2) {};
\LevelI{B1}
\vitalnode{\boxwidth}{4\innerboxheight}{B2}{$N_2$}

\node [box] at (C) {};
\node [box, minimum height=\innerboxheight] at (C1) {};
\node [box, minimum height=3\innerboxheight] at (C2) {};
\node [box, minimum height=\innerboxheight] at (C5) {};
\LevelI{C1}
\vitalnode{\boxwidth}{3\innerboxheight}{C2}{$N_3$}
\LevelV{C5}

\node [box] at (D) {};
\node [box, minimum height=\innerboxheight] at (D1) {};
\node [box, minimum height=2\innerboxheight] at (D2) {};
\node [box, minimum height=\innerboxheight] at (D4) {};
\node [box, minimum height=\innerboxheight] at (D5) {};
\LevelI{D1}
\vitalnode{\boxwidth}{2\innerboxheight}{D2}{$N_4$}
\LevelIV{D4}
\LevelV{D5}

\node [box] at (E) {};
\node [box, minimum height=\innerboxheight] at (E1) {};
\node [box, minimum height=\innerboxheight] at (E2) {};
\node [box, minimum height=\innerboxheight] at (E3) {};
\node [box, minimum height=\innerboxheight] at (E4) {};
\node [box, minimum height=\innerboxheight] at (E5) {};
\LevelI{E1}
\LevelII{E2}
\LevelIII{E3}
\LevelIV{E4}
\LevelV{E5}

\coordinate (a1) at ($(A)-(0,\treeshift)$);

\coordinate (b11) at ($(B)-(0,\treeshift-2.2\vertexdist)$);
\coordinate (b12) at ($(B)-(0,\treeshift-\vertexdist)$);
\coordinate (b13) at ($(B)-(0,\treeshift)$);
\coordinate (b2) at ($(B)-(0,\treeshift+\vertexdist)$);

\coordinate (c2) at ($(C)-(0,\treeshift)$);
\coordinate (c51) at ($(C)-(0,\treeshift+\vertexdist)$);
\coordinate (c52) at ($(C)-(0,\treeshift+2\vertexdist)$);
\coordinate (c53) at ($(C)-(0,\treeshift+3.2\vertexdist)$);

\coordinate (d2) at ($(D)-(0,\treeshift-\vertexdist/2)$);
\coordinate (d41) at ($(D)-(0,\treeshift+\vertexdist/2)$);
\coordinate (d42) at ($(D)-(0,\treeshift+3\vertexdist/2)$);

\coordinate (e21) at ($(E)-(0,\treeshift-1.5\vertexdist)$);
\coordinate (e22) at ($(E)-(0,\treeshift-0.5\vertexdist)$);
\coordinate (e31) at ($(E)-(0,\treeshift+0.5\vertexdist)$);
\coordinate (e32) at ($(E)-(0,\treeshift+2\vertexdist)$);
\coordinate (e33) at ($(E)-(0,\treeshift+3\vertexdist)$);

\path[edge] (a1) -- (b11);
\path[edge] (a1) -- (b12);
\path[edge] (a1) -- (b13);
\path[edge] (a1) -- (b2);

\path[edge] (b2) -- (c2);
\path[edge] (b2) -- (c51);
\path[edge] (b2) -- (c52);
\path[edge] (b2) -- (c53);

\path[edge] (c2) -- (d2);
\path[edge] (c2) -- (d41);
\path[edge] (c2) -- (d42);

\path[edge] (d2) -- (e21);
\path[edge] (d2) -- (e22);
\path[edge] (d2) -- (e31);
\path[edge] (d2) -- (e32);
\path[edge] (d2) -- (e33);

\draw (a1) node[bvertex] {} node[left=0.13] (v1) {$N_1$};

\draw (b11) node[bvertex] {};
\draw (b12) node[wvertex] {} node[right=0.13] (f1) {$F_1$};
\draw (b13) node[wvertex] {} node[right=0.13] {$F_2$};
\draw (b2) node[bvertex] {} node[below left=0.13] (v2) {$N_2$};

\draw (c2) node[bvertex] {} node[above left=0.13] {$N_3$};
\draw (c51) node[wvertex] {} node[right=0.13] {$F_3$};
\draw (c52) node[wvertex] {} node[right=0.13] (f4) {$F_4$};
\draw (c53) node[bvertex] {};

\draw (d2) node[bvertex] {} node[above left=0.13] {$N_4$};
\draw (d41) node[wvertex] {} node[right=0.13] {$F_5$};
\draw (d42) node[wvertex] {} node[right=0.13] (f6) {$F_6$};

\draw (e21) node[wvertex] {} node[right=0.13] (f7) {$F_7$};
\draw (e22) node[wvertex] {} node[right=0.13] {$F_8$};
\draw (e31) node[wvertex] {} node[right=0.13] (f9) {$F_9$};
\draw (e32) node[bvertex] {};
\draw (e33) node[bvertex] {};

\draw ($(a1)+(-0.5\vertexdist,-1.9\vertexdist)$) node (DN) {} node[above right=-0.5cm of DN] {\begin{Large}$D(N)$\end{Large}};

\draw[dashed] \convexpath{v1,b12,f7,f9,f6,f4,DN}{0.45cm};

\coordinate (l1) at ($(a1)+(-\boxwidth/2+\boxwidth/6,-4.2\vertexdist)$);

\node [active, minimum width=\boxwidth/3-\innermargin, minimum height=\innerboxheight-\innermargin, rounded corners=0.1\boxwidth/3] at (l1) {}
node[right = \boxwidth/6 of l1] (l2) {-- active}
node [right = \boxwidth/3 of l2, vital, minimum width=\boxwidth/3-\innermargin, minimum height=\innerboxheight-\innermargin, rounded corners=0.1\boxwidth/3] (l3) {}
node[right = 0.1cm of l3] (l4) {-- $Q$-node}
node [right = \boxwidth/3 of l4, nonvital, minimum width=\boxwidth/3-\innermargin, minimum height=\innerboxheight-\innermargin, rounded corners=0.1\boxwidth/3] (l5) {}
node[right = 0.1cm of l5] (l6) {-- recursive or problematic};

\end{tikzpicture}
\end{center}
\caption{\label{fig:Q_path}
The path $Q=\{N=N_1,N_2,N_3,N_4\}$ and the sets $A(Q) = \{F_1,F_2\}$ and $B(Q) = \{F_3,F_4,F_5,F_6\}$.
The nodes $\{F_7,F_8,F_9\}$ are the children of the last node $N_4$ in the path $Q$.
Every $F_i$ is a recursive or a problematic child of a node from $Q$.}
\end{figure}
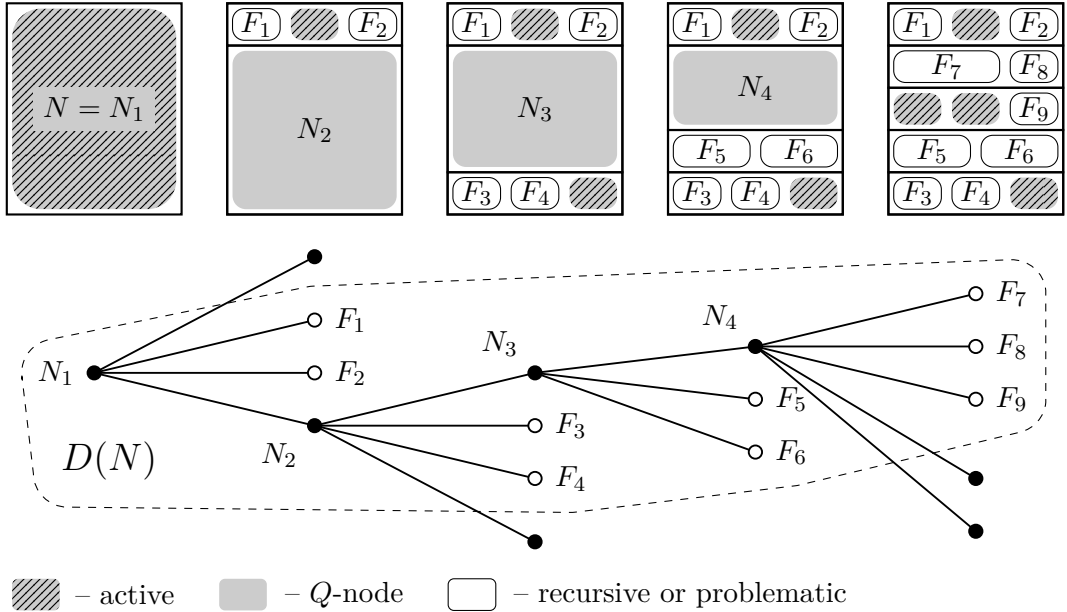

Suppose $N$ is a complete bipartite poset of characteristics $(u,\infty)$.
Since a complete bipartite poset may have two complete bipartite posets of the same width as its children,
the set $Q$ of all nodes in $D(N)$ with characteristic $(u,\infty)$ may form a subtree of $(\mathcal{N}, {\subnode})$ rooted in $N$.
Suppose $\mathcal{L}$ is a set of all leaves of the tree $Q$.

Similarly to the previous case, the algorithm maintains two sets of nodes $A(R)$ and $B(R)$, 
for every root-to-leaf path $R$ in $Q$.
This time the sets in $\{A(R) , B(R): R \text{ is a root to leaf path in } Q \}$
form a partition of all recursive and problematic children of the nodes from $Q \setminus \mathcal{L}$.
The sets $A(R)$ and $B(R)$ are processed as follows.
At the time $N$ appears in the node tree, i.e., when $Q = \{N\}$ and $Q$ has a single root-to-leaf path $R=\{N\}$, the sets
$A(R)$ and $B(R)$ are set empty.
Now suppose $R$ is a root-to-leaf path in $Q$ and a node, say $N'$, from $R$ is being split.
When $N'$ has exactly one child with characteristics $(u,\infty)$, the sets $A(R)$ and $B(R)$ 
are extended as in the previous case.
When $N'$ has no children with characteristics $(u,\infty)$, the construction of $A(R)$ and $B(R)$ is finished.
Eventually, when $N'$ splits into two complete bipartite posets of characteristics $(u,\infty)$, 
an upper-layered $M'$ and a lower-layered $M''$, 
the tree $Q$ is expanded by the nodes $M',M''$ attached to the node $N'$, 
and the root-to-leaf path $R$ in $Q$ is replaced by two root-to-leaf paths, $R'$ from $N$ to $M'$ and $R''$ from $N$ to $M''$.
Eventually, the sets $A(R'),B(R'),A(R''),B(R'')$ are set such that:
$$A(R') = A(R), B(R') = \emptyset,\ A(R'') = \emptyset, B(R'') = B(R).$$
The following proposition is analogous to the previous one and admits a similar proof:
\begin{proposition} 
\label{prop:A(Q)B(Q)-tree-online}
The following sentences hold:
\begin{enumerate} 
\item \label{prop:A(Q)B(Q)-tree-online-1}
For every root-to-leaf path $R$ in $Q$, the edge posets $(A(R)_E,{\leq_E})$ and $(B(R)_E,{\leq_E})$ are down-growing and up-growing orders of width at most $w^3$, respectively.
\item \label{prop:A(Q)B(Q)-tree-online-2}
For every two different root-to-leaf paths $R'$ and $R''$ in $Q$, we have either $A(R')_E \leq_E A(R'')_E$ or $A(R'')_E \leq_E A(R')_E$
and either $B(R')_E \leq_E B(R'')$ or $B(R'') \leq_E B(R')$.
\end{enumerate}
\end{proposition}

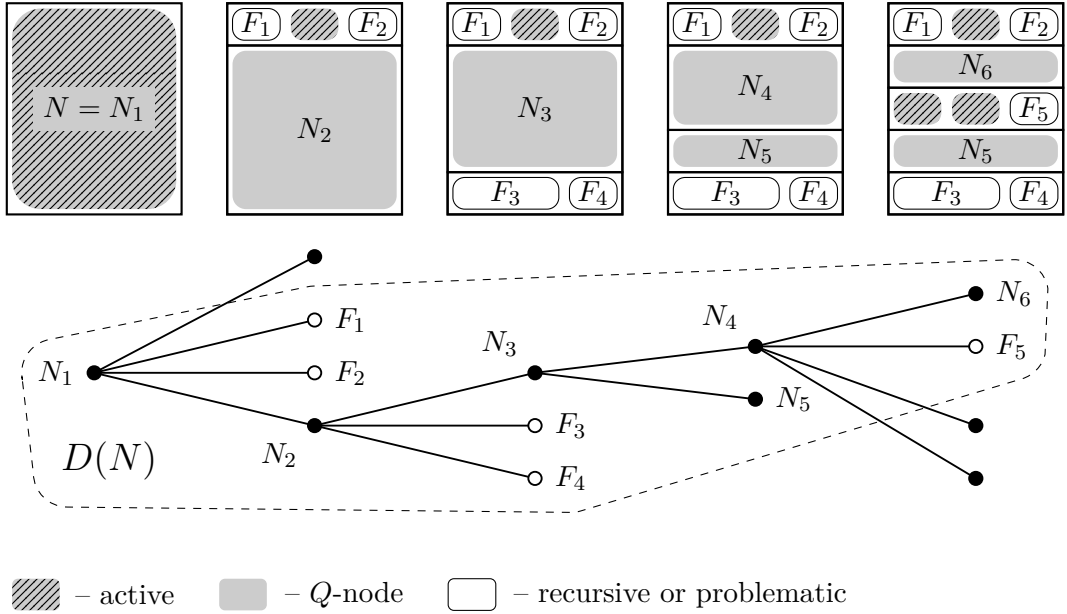
\begin{figure}[htbp!]
\begin{center}
\begin{tikzpicture}

\coordinate (A) at (2,2);
\coordinate (B) at ($(A)+(\boxwidth+\boxshift,0)$);
\coordinate (C) at ($(B)+(\boxwidth+\boxshift,0)$);
\coordinate (D) at ($(C)+(\boxwidth+\boxshift,0)$);
\coordinate (E) at ($(D)+(\boxwidth+\boxshift,0)$);

\coordinate (B1) at ($(B)+(0,\innerboxheight*2)$);
\coordinate (B2) at ($(B)-(0,\innerboxheight/2)$);

\coordinate (C1) at ($(C)+(0,\innerboxheight*2)$);
\coordinate (C2) at (C);
\coordinate (C5) at ($(C)-(0,\innerboxheight*2)$);

\coordinate (D1) at ($(D)+(0,\innerboxheight*2)$);
\coordinate (D2) at ($(D)+(0,\innerboxheight/2)$);
\coordinate (D4) at ($(D)-(0,\innerboxheight)$);
\coordinate (D5) at ($(D)-(0,\innerboxheight*2)$);

\coordinate (E1) at ($(E)+(0,\innerboxheight*2)$);
\coordinate (E2) at ($(E)+(0,\innerboxheight)$);
\coordinate (E3) at (E);
\coordinate (E4) at ($(E)-(0,\innerboxheight)$);
\coordinate (E5) at ($(E)-(0,\innerboxheight*2)$);

\newcommand{\LevelI}[1]{
\coordinate (B11) at ($(#1)-(\boxwidth/3,0)$);
\coordinate (B12) at (#1);
\coordinate (B13) at ($(#1)+(\boxwidth/3,0)$);
\nonvitalnode{\boxwidth/3}{\innerboxheight}{B11}{$F_1$}
\activenodee{\boxwidth/3}{\innerboxheight}{B12}
\nonvitalnode{\boxwidth/3}{\innerboxheight}{B13}{$F_2$}}

\newcommand{\LevelII}[1]{
\vitalnode{\boxwidth}{\innerboxheight}{#1}{$N_6$}}
\newcommand{\LevelIII}[1]{
\coordinate (E31) at ($(#1)-(\boxwidth/3,0)$);
\coordinate (E32) at (#1);
\coordinate (E33) at ($(#1)+(\boxwidth/3,0)$);
\activenodee{\boxwidth/3}{\innerboxheight}{E31}
\activenodee{\boxwidth/3}{\innerboxheight}{E32}
\nonvitalnode{\boxwidth/3}{\innerboxheight}{E33}{$F_5$}}

\newcommand{\LevelIV}[1]{
\vitalnode{\boxwidth}{\innerboxheight}{#1}{$N_5$}}

\newcommand{\LevelV}[1]{
\coordinate (C51) at ($(#1)-(\boxwidth/6,0)$);
\coordinate (C52) at ($(#1)+(\boxwidth/3,0)$);
\nonvitalnode{2\boxwidth/3}{\innerboxheight}{C51}{$F_3$}
\nonvitalnode{\boxwidth/3}{\innerboxheight}{C52}{$F_4$}}

\node [box] at (A) {};
\activenode{\boxwidth}{\boxheight}{A}{$N=N_1$}

\node [box] at (B) {};
\node [box, minimum height=\innerboxheight] at (B1) {};
\node [box, minimum height=4\innerboxheight] at (B2) {};
\LevelI{B1}
\vitalnode{\boxwidth}{4\innerboxheight}{B2}{$N_2$}

\node [box] at (C) {};
\node [box, minimum height=\innerboxheight] at (C1) {};
\node [box, minimum height=3\innerboxheight] at (C2) {};
\node [box, minimum height=\innerboxheight] at (C5) {};
\LevelI{C1}
\vitalnode{\boxwidth}{3\innerboxheight}{C2}{$N_3$}
\LevelV{C5}

\node [box] at (D) {};
\node [box, minimum height=\innerboxheight] at (D1) {};
\node [box, minimum height=2\innerboxheight] at (D2) {};
\node [box, minimum height=\innerboxheight] at (D4) {};
\node [box, minimum height=\innerboxheight] at (D5) {};
\LevelI{D1}
\vitalnode{\boxwidth}{2\innerboxheight}{D2}{$N_4$}
\LevelIV{D4}
\LevelV{D5}

\node [box] at (E) {};
\node [box, minimum height=\innerboxheight] at (E1) {};
\node [box, minimum height=\innerboxheight] at (E2) {};
\node [box, minimum height=\innerboxheight] at (E3) {};
\node [box, minimum height=\innerboxheight] at (E4) {};
\node [box, minimum height=\innerboxheight] at (E5) {};
\LevelI{E1}
\LevelII{E2}
\LevelIII{E3}
\LevelIV{E4}
\LevelV{E5}

\coordinate (a1) at ($(A)-(0,\treeshift)$);

\coordinate (b11) at ($(B)-(0,\treeshift-2.2\vertexdist)$);
\coordinate (b12) at ($(B)-(0,\treeshift-\vertexdist)$);
\coordinate (b13) at ($(B)-(0,\treeshift)$);
\coordinate (b2) at ($(B)-(0,\treeshift+\vertexdist)$);

\coordinate (c2) at ($(C)-(0,\treeshift)$);
\coordinate (c51) at ($(C)-(0,\treeshift+\vertexdist)$);
\coordinate (c52) at ($(C)-(0,\treeshift+2\vertexdist)$);

\coordinate (d2) at ($(D)-(0,\treeshift-\vertexdist/2)$);
\coordinate (d4) at ($(D)-(0,\treeshift+\vertexdist/2)$);

\coordinate (e21) at ($(E)-(0,\treeshift-1.5\vertexdist)$);
\coordinate (e31) at ($(E)-(0,\treeshift-0.5\vertexdist)$);
\coordinate (e32) at ($(E)-(0,\treeshift+1\vertexdist)$);
\coordinate (e33) at ($(E)-(0,\treeshift+2\vertexdist)$);

\path[edge] (a1) -- (b11);
\path[edge] (a1) -- (b12);
\path[edge] (a1) -- (b13);
\path[edge] (a1) -- (b2);

\path[edge] (b2) -- (c2);
\path[edge] (b2) -- (c51);
\path[edge] (b2) -- (c52);

\path[edge] (c2) -- (d2);
\path[edge] (c2) -- (d4);

\path[edge] (d2) -- (e21);
\path[edge] (d2) -- (e31);
\path[edge] (d2) -- (e32);
\path[edge] (d2) -- (e33);

\draw (a1) node[bvertex] {} node[left=0.13] (v1) {$N_1$};

\draw (b11) node[bvertex] {};
\draw (b12) node[wvertex] {} node[right=0.13] (f1) {$F_1$};
\draw (b13) node[wvertex] {} node[right=0.13] {$F_2$};
\draw (b2) node[bvertex] {} node[below left=0.13] (v2) {$N_2$};

\draw (c2) node[bvertex] {} node[above left=0.13] {$N_3$};
\draw (c51) node[wvertex] {} node[right=0.13] {$F_3$};
\draw (c52) node[wvertex] {} node[right=0.13] (f4) {$F_4$};

\draw (d2) node[bvertex] {} node[above left=0.13] {$N_4$};
\draw (d4) node[bvertex] {} node[right=0.13] (v5) {$N_5$};

\draw (e21) node[bvertex] {} node[right=0.13] (v6) {$N_6$};
\draw (e31) node[wvertex] {} node[right=0.13] (f5) {$F_5$};
\draw (e32) node[bvertex] {};
\draw (e33) node[bvertex] {};

\draw ($(a1)+(-0.5\vertexdist,-1.9\vertexdist)$) node (DN) {} node[above right=-0.5cm of DN] {\begin{Large}$D(N)$\end{Large}};

\draw[dashed] \convexpath{v1,b12,v6,f5,f4,DN}{0.45cm};

\coordinate (l1) at ($(a1)+(-\boxwidth/2+\boxwidth/6,-4.2\vertexdist)$);

\node [active, minimum width=\boxwidth/3-\innermargin, minimum height=\innerboxheight-\innermargin, rounded corners=0.1\boxwidth/3] at (l1) {}
node[right = \boxwidth/6 of l1] (l2) {-- active}
node [right = \boxwidth/3 of l2, vital, minimum width=\boxwidth/3-\innermargin, minimum height=\innerboxheight-\innermargin, rounded corners=0.1\boxwidth/3] (l3) {}
node[right = 0.1cm of l3] (l4) {-- $Q$-node}
node [right = \boxwidth/3 of l4, nonvital, minimum width=\boxwidth/3-\innermargin, minimum height=\innerboxheight-\innermargin, rounded corners=0.1\boxwidth/3] (l5) {}
node[right = 0.1cm of l5] (l6) {-- recursive or problematic};

\end{tikzpicture}
\end{center}
\caption{The tree $Q$ consists of two root-to-leaf paths: $R'=\{N_1,N_2,N_3,N_4,N_6\}$ and $R''= \{N_1,N_2,N_3,N_5\}$.
The sets $A(R')=\{F_1,F_2\}$, $B(R')=\{F_5\}$, $A(R'')=\emptyset$, $B(R'') = \{F_3,F_4\}$.}
\end{figure}

For every root-to-leaf path $R$ in $Q$, Felsner's algorithm is used to partition each of the edge posets $(A(R)_E,{\leq_E})$ and $(B(R)_E,{\leq_E})$ into at most $\binom{w^3+1}{2}$ chains.
By Proposition \ref{prop:A(Q)B(Q)-tree-online}.\eqref{prop:A(Q)B(Q)-tree-online-2}, 
the edges from the set $\bigcup \{A(R)_E \cup B(R)_E : R \text{ is a root-to-leaf path in } Q\}$ can be partitioned into at most $2\binom{w^3+1}{2}$ edge chains.
Using the replacing-chains-by-colors method, if we assign a bundle of $2\binom{w^3+1}{2}$ colors for every edge of $N$, with the use of these colors we can color the edges from the set $\bigcup \{A(R)_E \cup B(R)_E : R \text{ is a root-to-leaf path in } Q\}$ such that \eqref{eq:prop_edge_coloring}-property holds.

We are left with the coloring of the nodes in the tree $Q$ and possible non-active children of the leaf nodes of the tree $Q$.
Similarly to the previous case, we reserve two additional colors in a bundle of every edge of $N$ for this purpose.
The nodes of $Q$ are colored as follows. 
If a node $N'$ in $Q$ has one child $M'$ in $Q$, we project additional colors
from $N'$ to $M'$.
If a node $N'$ in $Q$ has two children $M'$ and $M''$ in $Q$, we use the `shuffling-colors' method to color
$M'$ and $M''$.
Hence, every edge from both $M'$ and $M''$ is covered by two additional colors.
Eventually, the recursive and the non-active children of the leaf nodes in $Q$ are colored the same way as in the previous case.

We can summarize the result of this part as:
\begin{proposition}
\label{prop:coloring_Q_and_its_children}
If every edge of an active node $N$ with characteristics $(u,s)$ is assigned a bundle of $\lambda_2(w) = 2\binom{w^3+1}{2}+2$ colors, then, with the use of these colors, the algorithm can color the nodes from~$Q$ and their recursive and problematic children such that
\eqref{eq:prop_edge_coloring} property is kept. 
\end{proposition}
To color the descendants of a recursive or a problematic node $M$, 
we need to have, respectively, $\lambda(\lfloor \sqrt{w}\rfloor)$ and $\lambda_3(w)$ colors on every edge of $M$.
To achieve this goal, we replace every color used for coloring the edges of $N$
by a bundle of $\lambda_3(w)$ colors (note that $\lambda_3(w) \geq \lambda(\lfloor \sqrt{w} \rfloor)$).
Having in mind Proposition~\ref{prop:coloring_Q_and_its_children}, this requires to have $\lambda_2(w)\lambda_3(w)$ 
colors on every edge of $N$, as we have guaranteed.

\subsubsection{Coloring recursive nodes}
\label{subsec:recursive_nodes}
Whenever the algorithm detects a recursive node $M$ in $D(N)$,
all descendants of $M$ (note that all of them are in $D(N)$) are colored by the recursive call of the algorithm coloring regular posets. 
Clearly, the descendants of $M$ can be colored recursively with the colors of the edges of $M$ 
provided each edge of $M$ is assigned a bundle of $\lambda(\lfloor \sqrt{w} \rfloor)$ colors -- see property \eqref{eq:P-property} of the algorithm.
Note that this condition was asserted for recursive children of the nodes from $Q$ as $\lambda(\lfloor \sqrt{w} \rfloor) \leq \lambda_3(w)$.

\subsubsection{Coloring problematic nodes}
\label{subsec:problematic_nodes}
Let $M$ be a problematic child of a node from $Q$ with characteristics $(u',s')$.
Since $M$ is in $D(N)$, $(u',s')$ is lexicographically smaller than $(u,s)$.
Suppose $D(M)$ is the set of all descendants of $M$ in $(\mathcal{N},{\subnode})$.
Since $M$ has no Dilworth's clique of size $\lceil \sqrt{w}\rceil$, all nodes in $D(M)$ contain no Dilworth's clique of size $\lceil \sqrt{w}\rceil$ as well, and hence they are either problematic or recursive.
In particular, $D(M) \subset D(N)$.
The following lemma will complete the description of the algorithm.
\begin{proposition}
If every problematic child $M$ of a node from $Q$ is assigned a bundle of\\
$\lambda_3(w) = w^4\big(2\binom{w^3+1}{2}+2\big)\lambda(\lfloor \sqrt{w} \rfloor)$
colors, then, with the use of these colors, the algorithm can color all nodes in $D(M)$ 
such that \eqref{eq:prop_edge_coloring} property is kept and every edge of a recursive node in $D(M)$ is assigned a bundle of $\lambda(\lfloor \sqrt{w} \rfloor)$ colors.
\end{proposition}

We split all problematic nodes $K$ in $D(M)$ into two classes:
\begin{itemize}
 \item $K$ is a \emph{first problematic} node in $D(M)$ if it is the first problematic node 
 with characteristics $(\width(K),\surplus(K))$ on the path from $M$ to $K$,
 \item $K$ is a \emph{subsequent problematic} node in $D(M)$ if $K$ is a child of some problematic node from $D(M)$ with the same characteristics as $K$.
\end{itemize}
In particular, note that $M$ is a first problematic node in $D(M)$.

Whenever a first problematic node $K$ of characteristics $(u'',s'')$ appears in the node tree, 
the algorithm reserves a bundle of $\big(2\binom{w^3+1}{2}+2\big)\lambda(\lfloor \sqrt{w} \rfloor)$ colors 
to every edge of $M$ and projects these colors on the edges of $K$.
Then, these colors are used to color:
\begin{itemize}
 \item all descendants of $K$ with characteristics $(u'',s'')$ in such a way that each edge from such a node receives a bundle of $2\lambda(\lfloor \sqrt{w} \rfloor)$ colors,
 \item all recursive children of the nodes from a path consisting of all descendants of $K$ with characteristics $(u'',s'')$ in such a way that each edge from such a node receives a bundle of $\lambda(\lfloor \sqrt{w} \rfloor)$ colors.
\end{itemize}
Clearly, the algorithm uses the ideas described in the previous section to perform the above-mentioned tasks. 

There is one issue to be clarified.
Each edge of a node $M$ is colored with $w^4\big(2\binom{w^3+1}{2}+2\big)\lambda(\lfloor \sqrt{w} \rfloor)$ colors,
so the algorithm may project colors onto at most $w^4$ first problematic nodes in the set $D(M)$.
The next lemma shows that there are at most $w^2$ first problematic nodes in $D(M)$ with characteristics $(u'',s'')$.
Since there are at most $w^2$ possible characteristics, there is no more than $w^4$ first problematic nodes in $D(M)$,
so we have enough colors to color all of them.

\begin{lemma}
\label{lem:problematic_nodes}
The set $\mathcal{F}$ of all first problematic nodes in $D(M)$ with characteristics $(u'',s'')$ contains at most $(u')^2 \leq w^2$ nodes.
\end{lemma}

\begin{proof}
First, note that $\mathcal{F}$ is ancestor-free.
Consider a chain partition $C_1,\ldots,C_{u'}$ of $(\Int(M),{\leq})$ into $u'$ chains -- such a partition exists as $(\Int(M),{\leq})$ has width $u'$.
Suppose $M=(X,Y,{<})$, where $X=\{x_1,\ldots,x_{u'}\}$ and $Y=\{y_1,\ldots,y_{u'}\}$ are such that $x_i,y_i \in C_i$.
Consider a first problematic node $K=(Z,T,{<})$ of characteristics $(u'',s'')$.
Since $K$ is regular, $Z \cup T$ has a non-empty intersection with exactly $u''$ chains in the set $\{C_1,\ldots,C_w\}$.
Suppose $I_K$ is the set of all indices $i \in [u']$ such that $C_i$ intersects $Z \cup T$.
In particular, note that for every $i \in I_K$ the sets $Z \cap C_i$ and $T \cap C_i$ are singletons.
Assume that $Z = \{z_i: i \in I_K\}$ and $T = \{t_i: i \in I_K\}$, where $z_i,t_i$ are such that $\{z_i\} = Z \cap C_i$ and $\{t_i\} = T \cap C_i$.
If $x_i < z_j$ for all $i,j \in I_K$, then 
$x_i < y_j$ for all $i,j \in I_K$ as $z_j < y_j$ for all $j \in [u']$.
Hence, $\{x_i,y_i: i \in I_K\}$ is a Dilworth clique $D$ of size at least $\lceil \sqrt{w} \rceil$ in $M$ as $u'' \geq \lceil \sqrt{w} \rceil$ as the edges $(x_j < y_j)$ for $j \in [u']\setminus I_K$ induce a perfect matching in $(X \setminus D, Y \setminus D,{<})$.
So, there are $i,j \in I_K$ such that $x_i$ is incomparable with $z_j$ -- see Figure \ref{fig_problematic_node}.
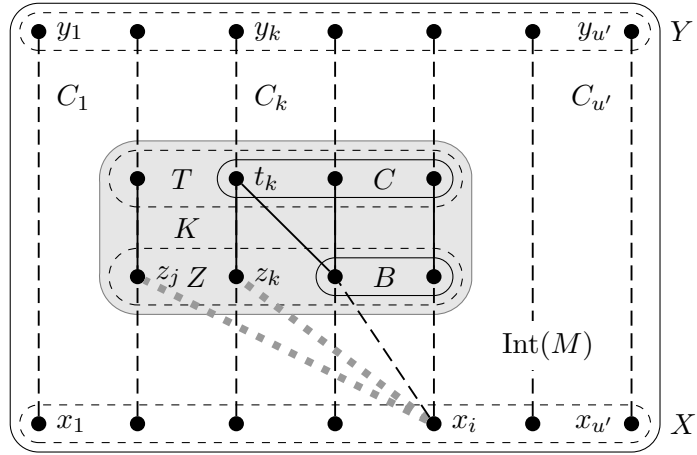
\begin{figure}[htbp!]
\begin{center}
\begin{tikzpicture}

\setlength{\hdist}{13mm}
\setlength{\vdist}{13mm}

\setlength{\antInner}{2.5mm}
\setlength{\antinner}{3.75mm}
\setlength{\nodeinner}{5mm}

\def\n{7}
\foreach \i in {1,...,\n} {
\coordinate (X\i) at (\i\hdist,0);
\coordinate (Z\i) at (\i\hdist,1.5\vdist);
\coordinate (T\i) at (\i\hdist,2.5\vdist);
\coordinate (Y\i) at (\i\hdist,4\vdist);
}

\draw[draw,dashed] \convexpath{X1,X\n}{\antInner};
\draw[draw,dashed] \convexpath{Y1,Y\n}{\antInner};
\draw[draw] \convexpath{X1,Y1,Y\n,X\n}{\antinner};
\draw[nodebordered] \convexpath{Z2,T2,T5,Z5}{\nodeinner};
\draw[draw,dashed] \convexpath{Z2,Z5}{\antinner};
\draw[draw,dashed] \convexpath{T2,T5}{\antinner};

\draw[draw] \convexpath{Z4,Z5}{\antInner};
\draw[draw] \convexpath{T3,T5}{\antInner};

\foreach \i in {1,...,\n} {
\node [bvertex] (x\i) at (X\i) {};
\node [bvertex] (y\i) at (Y\i) {};
}

\draw (Z2) node[bvertex] (z2) {} node[right=1mm] {$z_j$};
\draw (T2) node[bvertex] (t2) {};
\draw (Z3) node[bvertex] (z3) {} node[right=1mm] {$z_k$};
\draw (T3) node[bvertex] (t3) {} node[right=1mm] {$t_k$};
\draw (Z4) node[bvertex] (z4) {};
\draw (T4) node[bvertex] (t4) {};
\draw (Z5) node[bvertex] (z5) {};
\draw (T5) node[bvertex] (t5) {};

\foreach \i in {1,...,\n} {
\path[dedge] (x\i) -- (y\i);
}

\path[iedge] (x5) -- (z3);
\path[iedge] (x5) -- (z2);
\path[dedge] (x5) -- (z4);
\path[edge]  (z4) -- (t3);

\node [right=1mm] at (x1)  {$x_{1}$};
\node [right=1mm] at (x5)  {$x_{i}$};
\node [left=1mm]  at (x\n) {$x_{u'}$};
\node [right=1mm] at (y1)  {$y_{1}$};
\node [right=1mm] at (y3)  {$y_{k}$};
\node [left=1mm] at (y\n) {$y_{u'}$};

\node [right=1mm] at (barycentric cs:T1=0.45,Y1=0.55) {$C_1$};
\node [right=1mm] at (barycentric cs:T3=0.45,Y3=0.55) {$C_k$};
\node [left=1mm]  at (barycentric cs:T\n=0.45,Y\n=0.55) {$C_{u'}$};

\node [right] at ($(x\n)+(\antinner,0)$) {$X$};
\node [right] at ($(y\n)+(\antinner,0)$) {$Y$};
\node [fill=white] at ($(x6)+(2mm,0.8\vdist)$) {$\Int(M)$};

\path[edge] (z2) -- (t2);
\path[edge] (z3) -- (t3);
\path[edge] (z4) -- (t4);
\path[edge] (z5) -- (t5);

\draw (barycentric cs:z2=0.4,z3=0.6) node {$Z$};
\draw (barycentric cs:t2=0.55,t3=0.45) node {$T$};
\draw (barycentric cs:z2=0.25,z3=0.25,t2=0.25,t3=0.25) node {$K$};

\draw (barycentric cs:z4=0.5,z5=0.5) node {$B$};
\draw (barycentric cs:t4=0.5,t5=0.5) node {$C$};

\end{tikzpicture}
\end{center}
\caption{Node $K$ witnesses the comparability $x_i < y_k$.}
\label{fig_problematic_node}
\end{figure}
Consider the set $B = x_i \upseteq \cap Z$.
Since $z_i$ is in $B$ and $z_j$ is not in $B$, $B$ is a non-empty set strictly contained in $Z$. 
Since the surplus of $K$ is at least $1$, the set $C=B\upseteq \cap T$ has at least $|B|+1$ elements.
So, there is $k \in I_K$ such that $x_i$ and $z_k$ are incomparable but $x_i < t_k$.
Consequently, we have also $x_i < y_k$ as $t_k < y_k$. 
If this is the case, we say that $K$ \emph{witnesses the comparability $x_i < y_k$}.
Now, note that every two nodes from $\mathcal{F}$ witness different comparabilities
$(x_i < y_k)$ as $\mathcal{F}$ is ancestor-free.
Assuming we have $(u')^2$ nodes in $\mathcal{F}$, we conclude $M$ is a complete bipartite poset, 
which is impossible for a problematic node.
\end{proof}

\section{Summary}
\label{sec:summary}
The algorithm presented in this work extends the ideas presented in \cite{BK15}.
Both of these algorithms use the reduction from the on-line chain partitioning problem to the on-line chain partitioning of regular posets problem (Lemma \ref{lem:reg})
and both are based on nodes and their combinatorial properties in the node tree. 
In particular, they color active nodes in the same way.
However, the basic difference lies in the definition of an `active node':
here we assume an active node needs to have a Dilworth clique of size $\lceil \sqrt{w} \rceil$,
while in \cite{BK15} an active node needed to have a Dilworth clique of size $2$.
The approach adopted here leads to problems, which were not encountered in the previous work.
This includes, in particular, coloring problematic nodes in $D(M)$, which is solved by using recursion 
for coloring recursive nodes and by providing a polynomial upper bound on the number of first problematic nodes in the set $D(M)$ (Lemma \ref{lem:problematic_nodes}).

The algorithm presented in this work needs to test whether a bipartite poset of width $w$ has a Dilworth clique of width $\lceil \sqrt{w} \rceil$,
which is as hard as testing whether a bipartite graph with bipartition classes of size $w$ contains
a clique of size $\lceil \sqrt{w} \rceil$.
Johnson \cite{Joh87a} showed that such a problem is NP-complete, 
so we can not hope that our algorithm will work in polynomial time in the size of the presented poset.
Nevertheless, one can implement our algorithm such that it works in time $w^{O(\sqrt{w})}n$, where
$w$ is the width and $n$ is the size of the presented poset.

\section{Acknowledgment}
We thank anonymous reviewers for his/her thorough review and highly appreciate the comments and suggestions, which significantly contributed to improving the quality of the paper.

\bibliographystyle{plain}

\end{document}